\documentclass{sigplanconf}

\pdfoutput=1

\usepackage{graphicx}

\usepackage{amsthm,amssymb,amsmath,latexsym}

\usepackage{upgreek}

\usepackage{multirow}

\usepackage{url}

\usepackage[table]{xcolor}

\usepackage[all]{xy}

\usepackage{stmaryrd}

\usepackage{code}

\newcommand{\defterm}[1]{\textbf{#1}}

\newcommand{\eg}{\emph{e.g.}}

\newcommand{\syn}[1]{\mathsf{#1}}
\newcommand{\var}[1]{\mathit{#1}}
\newcommand{\s}[1]{\mathit{#1}}

\newcommand{\parto}{\rightharpoonup}
\newcommand{\dom}{\var{dom}}
\newcommand{\range}{\var{range}}

\newcommand{\compose}{\mathrel{\circ}}

\newcommand{\set}[1]{\left\{#1\right\}}
\newcommand{\setbuild}[2]{\left\{ #1 : #2\right\}}

\newcommand{\Pow}[1]{{\mathcal{P}\left(#1\right)}}
\newcommand{\PowSm}[1]{{\mathcal{P}(#1)}}

\newcommand{\union}{\cup}

\newcommand{\Union}{\bigcup}

\newcommand{\vect}[1]{\langle #1\rangle}
\newcommand{\vecp}[1]{\vec{#1}\;'}

\newcommand{\To}{\mathrel{\Rightarrow}}

\newcommand{\wt}{\sqsubseteq}

\newcommand{\join}{\sqcup}

\newcommand{\bigjoin}{\bigsqcup}

\DeclareMathOperator{\lfp}{lfp}

\newcommand{\QStates}{Q}
\newcommand{\FStates}{F}

\newcommand{\StackAlpha}{\Gamma}

\newcommand{\stackchar}{\gamma}

\newcommand{\sembr}[1]{\ensuremath{[\![{#1}]\!]}}

\newcommand{\opor}{\mathrel{|}}

\newcommand{\Alphabet}{A}

\newcommand{\produces}{\mathrel{::=}}

\newcommand{\vv}{v}

\newcommand{\lam}{\ensuremath{\var{lam}}}

\newcommand{\lamterm}{$\lambda$-term}
\newcommand{\lc}{$\lambda$-calculus}

\newcommand{\call}{\ensuremath{\var{call}}}

\newcommand{\lt}[2]{\lambda #1.#2}

\newcommand{\ttlp}{\mbox{\tt (}}
\newcommand{\ttrp}{\mbox{\tt )}}

\newcommand{\appform}[2]{\ttlp #1\; #2\ttrp}
\newcommand{\lamform}[2]{\ttlp \uplambda\;\ttlp#1\ttrp\;#2\ttrp}

\newcommand{\letiform}[3]{\ttlp {\tt let}\; \ttlp\ttlp#1\; #2\ttrp\ttrp\; #3\ttrp}

\newcommand{\fexpr}{f}
\newcommand{\expr}{e}
\newcommand{\aexpr}{\mbox{\sl {\ae}}}

\newcommand{\Eval}{{\mathcal{E}}}
\newcommand{\ArgEval}{{\mathcal{A}}}

\newcommand{\Inject}{{\mathcal{I}}}

\newcommand{\qstate}{q}

\newcommand{\store}{\sigma}

\newcommand{\env}{\rho}

\newcommand{\clo}{\var{clo}}
\newcommand{\cont}{\kappa}

\newcommand{\alloc}{\mathit{alloc}}

\newcommand{\addr}{a}

\newcommand{\aTo}{\leadsto}

\newcommand{\aInject}{{\hat{\mathcal{I}}}}

\newcommand{\sa}[1]{\widehat{\mathit{#1}}}

\newcommand{\aEval}{{\hat{\mathcal{E}}}}
\newcommand{\aArgEval}{{\hat{\mathcal{A}}}}

\newcommand{\astore}{{\hat{\sigma}}}
\newcommand{\aenv}{{\hat{\rho}}}

\newcommand{\aclo}{{\widehat{\var{clo}}}}

\newcommand{\acont}{{\hat{\kappa}}}

\newcommand{\aaddr}{{\hat{\addr}}}
\newcommand{\aalloc}{{\widehat{alloc}}}

\newcommand{\absmap}{\alpha}
\newcommand{\abs}[1]{|#1|}

\newcommand{\goodsingl}[1]{{#1}^{\checkmark}}

\definecolor{light-gray}{gray}{0.88}

\newcommand{\ControlStates}{Q}
\newcommand{\transfunction}{\delta}

\newcommand{\conf}{c}
\newcommand{\aconf}{{\hat c}}

\newcommand{\phrame}{\phi}
\newcommand{\aphrame}{\hat{\phi}}

\newcommand{\stackact}{g}

\DeclareMathOperator*{\PDTrans}{\longmapsto}

\newcommand{\fDSG}{\mathcal{DSG}}

\newcommand{\afPDS}{\widehat{\mathcal{PDS}}}

\newcommand{\afIPDS}{\widehat{\mathcal{IPDS}}}

\renewcommand{\Alphabet}{\Sigma}

\DeclareMathOperator*{\pdedge}{\rightarrowtail}

\newcommand{\triedge}[3]{#1 \mathrel{\pdedge^{#2}} #3}
\newcommand{\biedge}[2]{#1 \mathrel{\rightarrowtail} #2}

\newcommand{\DSStates}{S}
\newcommand{\DSEdges}{E}
\newcommand{\DSFrames}{\StackAlpha}
\newcommand{\dsframe}{\stackchar}
\newcommand{\dsstate}{s}

\newcommand{\mkDSG}{\mathcal{F}}

\newcommand{\fnet}[1]{\lfloor #1 \rfloor}
\newcommand{\fstackify}[1]{\lceil #1 \rceil}

\newcommand{\ecg}{$\epsilon$-closure graph}

\DeclareMathOperator*{\RPDTrans}{{\longmapsto\!\!\!\!\!\!\!\longrightarrow}}

\DeclareMathOperator*{\areaches}{\rightarrowtriangle}

\newcommand{\aCollect}{{\hat{G}}}

\newcommand{\mylongtitle}{Introspective Pushdown Analysis of Higher-Order Programs} 

\newcommand{\myshorttitle}{}
\newcommand{\mytitle}{\mylongtitle}

\newcommand{\mytitlebanner}{In submission}
\newcommand{\footertitle}{\myshorttitle}

\newtheorem{theorem}{Theorem}[section]

\begin{document}

\conferenceinfo{ICFP '12}{September 10--12, Copenhagen, Denmark.} 
\copyrightyear{2012} 
\copyrightdata{978-1-4503-1054-3/12/09} 

\titlebanner{\mytitlebanner}        
\preprintfooter{\footertitle}       

\title{\mytitle}


\authorinfo{Christopher Earl}
           {University of Utah}
           {cwearl@cs.utah.edu}

\authorinfo{Ilya Sergey}
           {KU Leuven}
           {ilya.sergey@cs.kuleuven.be}

\authorinfo{Matthew Might}
           {University of Utah}
           {might@cs.utah.edu}

\authorinfo{David Van Horn}
           {Northeastern University}
           {dvanhorn@ccs.neu.edu}

\maketitle

\begin{abstract}
  In the static analysis of functional programs, pushdown flow analysis and abstract garbage collection skirt just inside the boundaries of soundness and decidability.  Alone, each method  reduces analysis times and boosts precision by orders of magnitude.  This work illuminates and conquers the theoretical challenges that stand in the way of combining the power of these techniques.  The challenge in marrying these techniques is not subtle: computing the reachable control states of a pushdown system relies on limiting access during transition to the top of the stack; abstract garbage collection, on the other hand, needs full access to the entire stack to compute a root set, just as concrete collection does.  \emph{Introspective} pushdown systems resolve this conflict.  Introspective pushdown systems provide enough access to the stack to allow abstract garbage collection, but they remain restricted enough to compute control-state reachability, thereby enabling the sound and precise product of pushdown analysis and abstract garbage collection.  Experiments reveal synergistic interplay between the techniques, and the fusion demonstrates ``better-than-both-worlds'' precision.

\end{abstract}




\parskip=0.1in
\parindent=0in

\section{Introduction}

The recent development of a context-free\footnote{
 As in context-free language, not context-sensitivity.
} approach to control-flow analysis
(CFA2) by Vardoulakis and Shivers has provoked a seismic shift in the
static analysis of higher-order
programs~\cite{dvanhorn:DBLP:conf/esop/VardoulakisS10}.
Prior to CFA2, a precise analysis of recursive behavior 
had been a stumbling block---even though flow analyses have an important role to play in 
optimization for functional languages, such as 
flow-driven inlining~\cite{mattmight:Might:2006:DeltaCFA},
interprocedural constant propagation~\cite{mattmight:Shivers:1991:CFA}
and type-check
elimination~\cite{mattmight:Wright:1998:Polymorphic}.

While it had been possible to statically analyze
recursion \emph{soundly}, CFA2 made it possible to analyze recursion
\emph{precisely} by matching calls and returns without approximation.
In its pursuit of recursion,
clever engineering steered CFA2 just shy of undecidability.
The payoff is an order-of-magnitude reduction in 
analysis time and an order-of-magnitude increase
in precision.

For a visual measure of the impact, Figure~\ref{fig:diamond}
renders the abstract transition graph (a model of all possible traces through the program) for 
the toy program in Figure~\ref{fig:toy}.
For this example,
pushdown analysis
eliminates spurious return-flow from the
use of  recursion.
But, recursion is just one problem of many for flow analysis.
For instance, pushdown analysis still gets tripped up by
the 
spurious cross-flow problem;
at calls to \texttt{(id f)}
and \texttt{(id g)} in the previous example,
it thinks \texttt{(id g)} could be \texttt{f} \emph{or} \texttt{g}.

Powerful techniques such as abstract garbage collection~\cite{mattmight:Might:2006:GammaCFA}
were developed to solve the cross-flow problem.\footnote{
  The cross-flow problem arises because
  monotonicity prevents
  revoking a judgment like ``procedure ${\tt f}$ flows to {\tt x},''
  or ``procedure ${\tt g}$ flows to {\tt x},''
  once it's been made.
}
In fact, abstract garbage collection, by itself, also delivers orders-of-magnitude
improvements to analytic speed and precision.  
(See Figure~\ref{fig:diamond} again for a visualization of that impact.)

It is natural to ask: can abstract garbage collection and pushdown 
anlysis work together?
Can their strengths be multiplied?
At first glance, the answer appears to be a disheartening \emph{No}.

%


\begin{figure}
\begin{center}
\includegraphics[width=3in]{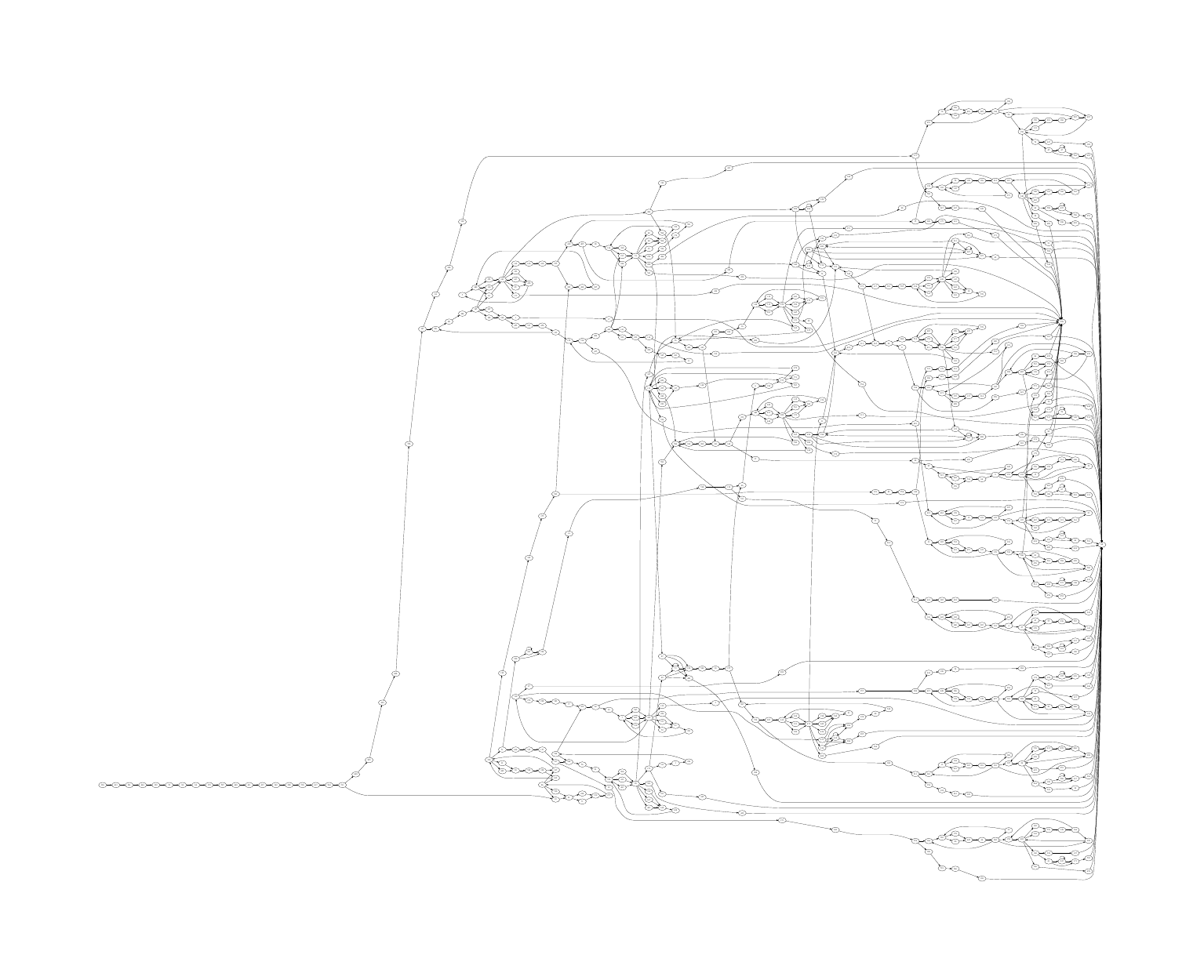}
\\
(1) without pushdown analysis or abstract GC: 653 states
\\
\includegraphics[width=3in]{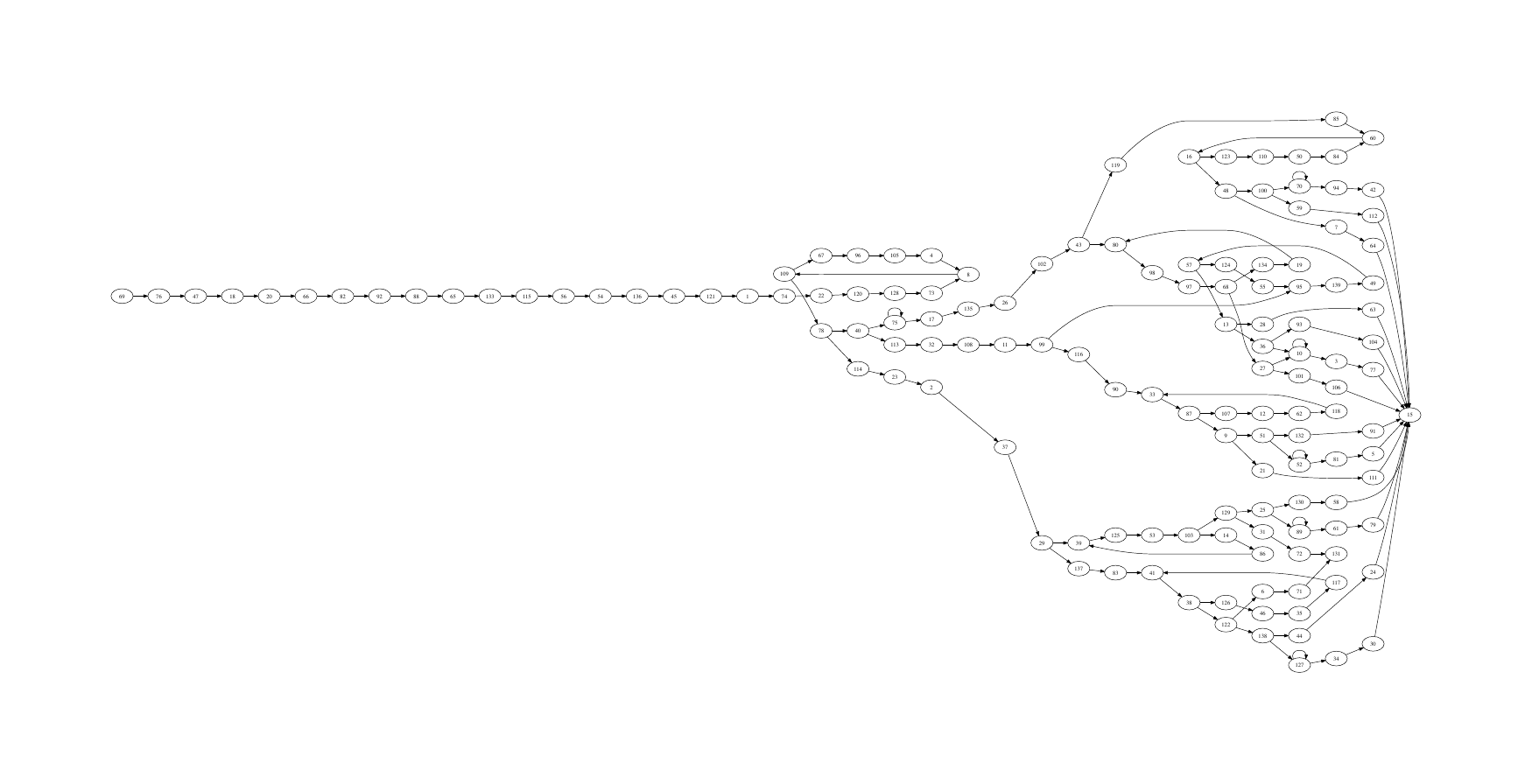}
\\
(2) with pushdown only: 139 states
\\
\includegraphics[width=3.3in]{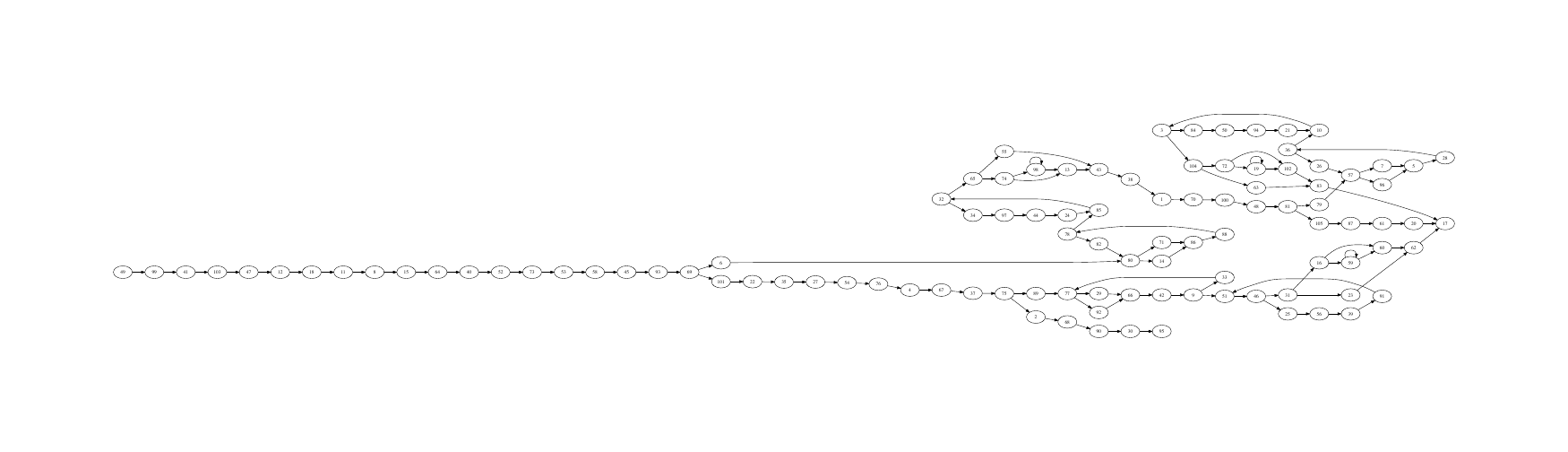}
\\
(3) with GC only: 105 states
\\
\includegraphics[width=2.5in]{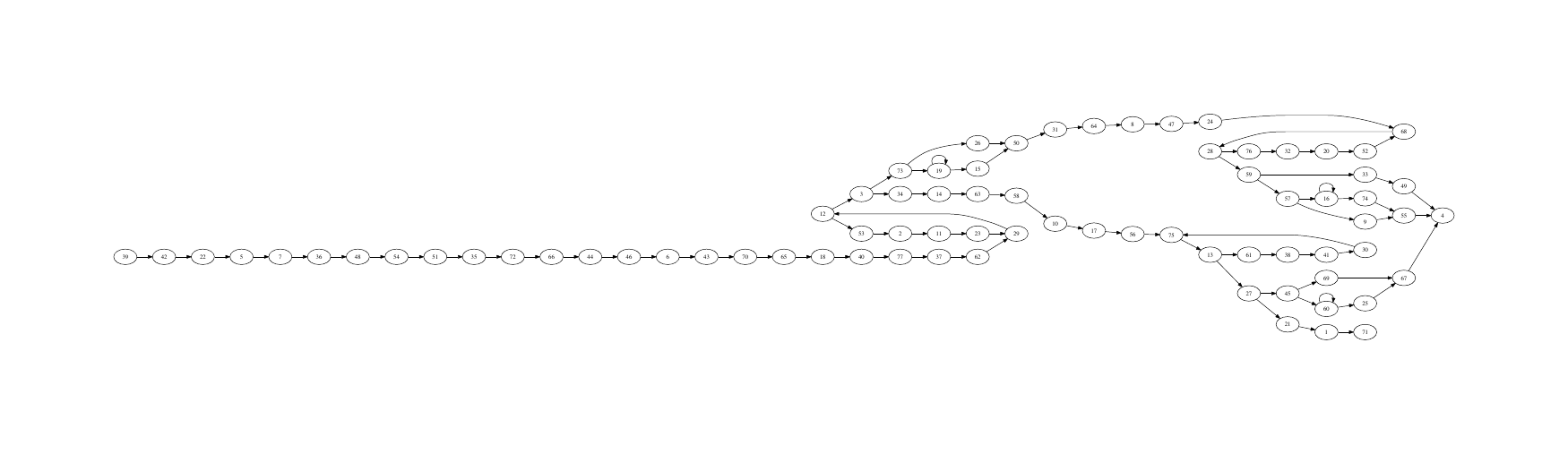}
\\
(4) with pushdown analysis and abstract GC: 77 states
\end{center}

\caption{
We generated an abstract transition graph
for the same program from Figure~\ref{fig:toy} four times: 
 (1) without pushdown analysis or abstract garbage collection;
 (2) with only abstract garbage collection;
 (3) with only pushdown analysis;
 (4) with both pushdown analysis and abstract garbage collection.
With only pushdown or abstract GC, the abstract transition graph shrinks
by an order of magnitude, but in different ways.
The pushdown-only analysis is confused by variables 
that are bound to several different higher-order functions,
but for short durations.
The abstract-GC-only is confused by non-tail-recursive loop structure.
With both techniques enabled, the graph shrinks by nearly half yet again
and fully recovers the control structure of the original program.}
\label{fig:diamond}
\end{figure}

\begin{figure}
\begin{code}
(define (id x) x)

(define (f n)
  (cond [(<= n 1)  1]
        [else      (* n (f (- n 1)))]))

(define (g n)
  (cond [(<= n 1)  1]
        [else      (+ (* n n) (g (- n 1)))]))
    
(print (+ ((id f) 3) ((id g) 4)))\end{code}
\caption{A small example to illuminate
the strengths and weaknesses of both
pushdown analysis and abstract garbage collection.}
\label{fig:toy}
\end{figure}

\subsection{The problem: The whole stack \emph{versus} just the top}

Abstract garbage collections seems to require more than pushdown analysis can
decidably provide: access to the full stack.
Abstract garbage collection, like its name implies, discards unreachable values
from an abstract store during the analysis.
Like concrete garbage collection, abstract garbage collection also begins its
sweep with a root set, and like concrete garbage collection, it must traverse
the abstract stack to compute that root set.
But, pushdown systems are restricted to viewing the top of the stack (or
a bounded depth)---a condition violated by this traversal.

Fortunately, abstract garbage collection does not need 
to arbitrarily modify the stack.
In fact, it does not even need to know the order of the frames;
it only needs the \emph{set} of frames on the stack.
We find a richer class of machine---\emph{introspective} pushdown
systems---which retains just enough restrictions to compute reachable 
control states, yet few enough to enable abstract garbage collection.

It is therefore possible to fuse the full benefits of abstract garbage
collection with pushdown analysis.
The dramatic reduction in abstract transition graph size
from the top to the bottom in Figure~\ref{fig:diamond}
(and echoed by later benchmarks) conveys the impact
of this fusion.

\paragraph{Secondary motivations}
There are three strong secondary motivations for this work:
(1) bringing context-sensitivity to pushdown analysis;
(2) exposing the context-freedom of the analysis; and
(3) enabling pushdown analysis
without continuation passing style.

In CFA2, monovariant (0CFA-like) context-sensitivity is etched directly into
the abstract semantics, which are in turn,
phrased in terms of an explicit (imperative) summarization algorithm
for a partitioned continuation-passing style.

In addition, the context-freedom of the analysis is buried implicitly
inside this algorithm.
No pushdown system or context-free grammar is explicitly identified.
A necessary precursor to our work was to make 
the pushdown system in CFA2 explicit.

A third motivation was to show that a transformation to continuation-passing
style is unnecessary for pushdown analysis.
In fact, pushdown analysis is arguably more natural
over direct-style programs.

\subsection{Overview}
We first review preliminaries to set a consistent feel for
terminology and notation, particularly with respect
to pushdown systems.
The derivation of the analysis 
begins with a concrete CESK-machine-style semantics for
A-Normal Form \lc{}.
The next step is an infinite-state abstract interpretation, constructed by
bounding the C(ontrol), E(nvironment) and S(tore) portions of the machine while
leaving the stack---the K(ontinuation)---unbounded.
A simple shift in perspective reveals that this abstract interpretation is a
rooted pushdown system.

We then introduce abstract garbage collection and quickly find
that it violates the pushdown model
with its traversals of the stack.
To prove the decidability of control-state reachability,
we formulate introspective pushdown systems, and 
recast abstract garbage collection within this framework.
We then show that control-state reachability is decidable for
introspective pushdown systems as well.

We conclude with an implementation and empirical evaluation that shows strong
synergies between pushdown analysis and abstract garbage collection, including
significant reductions in the size of the abstract state transition graph.

\subsection{Contributions}
We make the following contributions:
\begin{enumerate}
\item Our primary contribution is 
demonstrating the decidability of fusing
abstract garbage collection with pushdown flow analysis of higher-order
programs.
Proof comes in the form of a fixed-point solution for computing the reachable
control-states of an introspective pushdown system
and an embedding of abstract garbage collection
as an introspective pushdown system.

\item We show that classical notions of context-sensitivity, such as $k$-CFA
and poly/CFA, have direct generalizations in a pushdown setting:
monovariance\footnote{Monovariance refers to an abstraction that groups all bindings to the same variable together: there is \emph{one} abstract variant for all bindings to each variable.} is \emph{not} an essential restriction,
as in CFA2.

\item We make the context-free aspect of CFA2 explicit: we clearly define and
identify the pushdown system.
We do so by starting with a classical CESK machine and systematically
abstracting until a pushdown system emerges.
We also remove the orthogonal frame-local-bindings aspect of CFA2, so as to
directly solely on the pushdown nature of the analysis.

\item We remove the requirement for CPS-conversion
by synthesizing the analysis directly for direct-style (in
    the form of A-normal form lambda-calculus).

\item We empirically validate claims of improved
precision on a suite of benchmarks.
We find synergies between 
pushdown analysis 
and
abstract garbage collection 
that makes the whole greater that the sum of its parts.

\end{enumerate}

\section{Pushdown preliminaries}

The literature contains  many equivalent definitions of pushdown machines, so
we adapt our own definitions from Sipser~\cite{mattmight:Sipser:2005:Theory}.
\emph{Readers familiar with pushdown theory may wish to skip ahead.}

\subsection{Syntactic sugar}

When a triple $(x,\ell,x')$ is an edge in a labeled graph:
\begin{equation*}
x \pdedge^\ell x'  \equiv
(x,\ell,x')
\text.
\end{equation*}
Similarly, when a pair $(x,x')$ is a graph edge:
\begin{equation*}
\biedge{x}{x'} \equiv (x,x')
\text.
\end{equation*}
We use both
string and vector notation for sequences:
\begin{equation*}
a_1 a_2 \ldots a_n \equiv \vect{a_1,a_2,\ldots,a_n}
\equiv
\vec{a} 
\text.
\end{equation*}

\subsection{Stack actions, stack change and stack manipulation}

Stacks are sequences over a stack alphabet $\StackAlpha$.
To reason about stack manipulation concisely,
we first turn stack alphabets into ``stack-action'' sets;
each character represents a change to the stack: push, pop or no
change.

For each character $\stackchar$ in a stack alphabet $\StackAlpha$, the
\defterm{stack-action} set $\StackAlpha_\pm$ contains a push character
$\stackchar_{+}$; a pop character $\stackchar_{-}$;
and a no-stack-change indicator, $\epsilon$:
\begin{align*}
\stackact \in \StackAlpha_\pm &\produces \epsilon && \text{[stack unchanged]} 
\\
    &\;\;\opor\;\; \stackchar_{+}  \;\;\;\text{ for each } \stackchar \in \StackAlpha && \text{[pushed $\stackchar$]}
    \\
      &\;\;\opor\;\; \stackchar_{-}  \;\;\;\text{ for each } \stackchar \in \StackAlpha && \text{[popped $\stackchar$]}
      \text.
      \end{align*}
      In this paper, the symbol $\stackact$ represents some stack action.

      When we develop introspective pushdown systems, we are going
      to need formalisms for easily manipulating stack-action strings
      and stacks.
      Given a string of stack actions, we can compact it into a minimal
      string describing net stack change.
      We do so through the operator $\fnet{\cdot} : \StackAlpha_\pm^* \to
      \StackAlpha_\pm^*$, which cancels out opposing adjacent push-pop stack
      actions:
      \begin{align*}
      \fnet{\vec{\stackact} \; \stackchar_+\stackchar_- \; \vecp{\stackact}} &= 
      \fnet{\vec{\stackact} \; \vecp{\stackact}} 
      &
      \fnet{\vec{\stackact} \; \epsilon \; \vecp{\stackact}} &= 
      \fnet{\vec{\stackact} \; \vecp{\stackact}} 
      \text,
      \end{align*}
      so that
      $\fnet{\vec{\stackact}} = \vec{\stackact}\text,$
      if there are no cancellations to be made in the string $\vec{\stackact}$.

      We can convert a net string back into a stack by stripping off the
      push symbols with the stackify operator, $\fstackify{\cdot} :
      \StackAlpha^{*}_\pm \parto \StackAlpha^*$:
      \begin{align*}
      \fstackify{\stackchar_+ \stackchar_+' \ldots \stackchar_+^{(n)}} =
      \vect{\stackchar^{(n)}, \ldots, \stackchar', \stackchar}
      \text,
      \end{align*}
      and for convenience, $[\vec{\stackact}] = 
      \fstackify{\fnet{\vec{\stackact}}}$.
      Notice the stackify operator is defined for strings containing
      only push actions. 

        %
        %
        %

        \subsection{Pushdown systems}
        A \defterm{pushdown system} is a triple
        $M = (\ControlStates,\StackAlpha,\transfunction)$ where:
        \begin{enumerate}

        \item $\ControlStates$ is a finite set of control states;

        \item $\StackAlpha$ is a stack alphabet; and

        \item $\transfunction \subseteq
        \ControlStates \times \StackAlpha_\pm \times \ControlStates$ is a transition relation.
        \end{enumerate}
        The set $\ControlStates \times \StackAlpha^*$ is 
        called the \defterm{configuration-space} of this pushdown system.
        We use $\mathbb{PDS}$ to denote the class of all pushdown systems.
        \\

        \noindent
        For the following definitions, let $M = (\ControlStates,\StackAlpha,\transfunction)$.
        \begin{itemize}


        \item The labeled \defterm{transition relation} $(\PDTrans_{M}) \subseteq
        (\ControlStates \times \StackAlpha^*) \times 
        \StackAlpha_\pm \times 
        (\ControlStates \times \StackAlpha^*)$
        determines whether one configuration may transition to another while performing the given stack action:
        \begin{align*}
(\qstate, \vec{\stackchar}) 
  \mathrel{\PDTrans_M^\epsilon}
  (\qstate',\vec{\stackchar}) 
  & \text{ iff }
  \qstate \pdedge^\epsilon \qstate'
  \in \transfunction
  && \text{[no change]}
  \\
    (\qstate, \stackchar : \vec{\stackchar}) 
    \mathrel{\PDTrans_M^{\stackchar_{-}}}
    (\qstate',\vec{\stackchar})
    & \text{ iff }
    \qstate \pdedge^{\stackchar_{-}} \qstate'
    \in \transfunction
    && \text{[pop]}
    \\
      (\qstate, \vec{\stackchar}) 
      \mathrel{\PDTrans_{M}^{\stackchar_{+}}}
      (\qstate',\stackchar : \vec{\stackchar}) 
      & \text{ iff }
      \qstate \pdedge^{\stackchar_{+}} \qstate'
      \in \transfunction
      && \text{[push]}
      \text.
      \end{align*}

      \item If unlabelled, the transition relation $(\PDTrans)$ checks whether \emph{any} stack action can enable the transition:
      \begin{align*}
      \conf \mathrel{\PDTrans_{M}} \conf' \text{ iff }
      \conf \mathrel{\PDTrans_{M}^{\stackact}} \conf' \text{ for some stack action } \stackact 
      \text.
      \end{align*}

      \item

      For a string of stack actions $\stackact_1 \ldots
      \stackact_n$:
      \begin{equation*}
      \conf_0 \mathrel{\PDTrans_M^{\stackact_1\ldots\stackact_n}} \conf_n
      \text{ iff }
      \conf_0
      \mathrel{\PDTrans_M^{\stackact_1}} 
      \conf_1
      \mathrel{\PDTrans_M^{\stackact_2}} 
      \cdots
      \mathrel{\PDTrans_M^{\stackact_{n-1}}} 
      \conf_{n-1} 
      \mathrel{\PDTrans_M^{\stackact_n}} 
      \conf_n\text,
      \end{equation*}
      for some configurations $\conf_0,\ldots,\conf_n$.

      \item

      For the transitive closure:
      \begin{equation*}
      \conf \mathrel{\PDTrans_M^{*} \conf'
        \text{ iff }
        \conf \mathrel{\PDTrans_M^{\vec{\stackact}}} \conf'
          \text{ for some action string }
        \vec{\stackact}} 
        \text.
        \end{equation*}

        \end{itemize}

        %

        %

        \subsection{Rooted pushdown systems}

        A \defterm{rooted pushdown system} is a quadruple 
        $(\ControlStates,\StackAlpha,\transfunction,\qstate_0)$ in which 
        $(\ControlStates,\StackAlpha,\transfunction)$ is a pushdown system and
        $\qstate_0 \in \ControlStates$ is an initial (root) state.
        $\mathbb{RPDS}$  is the class of all rooted pushdown
        systems.

        For a rooted pushdown system $M =
        (\ControlStates,\StackAlpha,\transfunction,\qstate_0)$, we define 
        the \defterm{reachable-from-root transition relation}:
        \begin{equation*}
        \conf \RPDTrans_M^{\stackact} \conf' \text{ iff } 
(\qstate_0,\vect{})
  \mathrel{\PDTrans_M^*} 
  \conf 
  \text{ and }
  \conf 
  \mathrel{\PDTrans_M^\stackact}
  \conf'
  \text.
  \end{equation*}
  In other words, the root-reachable transition relation also makes
  sure that the root control state can actually reach the transition.

  We overload the root-reachable transition relation to operate on
  control states:
  \begin{equation*}
  \qstate 
  \mathrel{\RPDTrans_M^\stackact}
  \qstate' \text{ iff }
(\qstate,\vec{\stackchar}) 
  \mathrel{\RPDTrans_M^\stackact}
  (\qstate',\vecp{\stackchar}) 
  \text{ for some stacks }
  \vec{\stackchar},
  \vecp{\stackchar}
  \text.
  \end{equation*}
  For both root-reachable relations, if we elide the stack-action label,
  then, as in the un-rooted case, the transition holds if \emph{there
    exists} some stack action that enables the transition:
    \begin{align*}
    \qstate 
    \mathrel{\RPDTrans_M}
    \qstate' \text{ iff }
    \qstate 
    \mathrel{\RPDTrans_M^{\stackact}}
    \qstate' 
    \text{ for some action } \stackact
    \text.
    \end{align*}

    \subsection{Computing reachability in pushdown systems}

    A pushdown flow analysis 
    can be construed as
    computing the \emph{root-reachable}
    subset of control states in a rooted pushdown system, 
    $M = (\ControlStates,\StackAlpha,\transfunction,\qstate_0)$:
    \begin{align*}
    \setbuild{\qstate}{
      \qstate_0 
        \mathrel{\RPDTrans_M}
      \qstate
    }
\end{align*}
Reps~\emph{et. al} and many others provide a straightforward ``summarization'' algorithm
to compute this set~\cite{mattmight:Bouajjani:1997:PDA-Reachability,mattmight:Kodumal:2004:CFL,mattmight:Reps:1998:CFL,mattmight:Reps:2005:Weighted-PDA}.
Our preliminary report also offers a reachability algorithm 
tailored to higher-order
programs~\cite{mattmight:Earl:2010:Pushdown}.

\subsection{Nondeterministic finite automata}
In this work, we will need a finite description of 
all possible stacks at a given control state within
a rooted pushdown system.
We will exploit the fact that the set of stacks
at a given control point is a regular language.
Specifically, we will extract a nondeterministic finite automaton
accepting that language from the structure
of a rooted pushdown system.
A \defterm{nondeterministic finite automaton} (NFA) is a quintuple
$M = (\QStates, \Alphabet, \transfunction, \qstate_0, \FStates)$:
\begin{itemize}
\item $\ControlStates$ is a finite set of control states;

\item $\Alphabet$ is an input alphabet; 

\item $\transfunction 
  \subseteq
\ControlStates \times (\Alphabet \union \set{\epsilon}) 
  \times \ControlStates$ 
  is a transition relation.

  \item $\qstate_0$ is a distinguished start state.

  \item $\FStates \subseteq \QStates$ is a set of accepting states.
  \end{itemize}
  We denote the class of all NFAs as $\mathbb{NFA}$.

                  \section{Setting: A-Normal Form 
                    $\lambda$-calculus}
                    \label{sec:anf}

                    Since our goal is analysis of
                    \emph{higher-order languages}, we operate on the
                    \lc{}.
                    To simplify presentation of the concrete and abstract semantics, we choose 
                    A-Normal Form \lc{}.
                    (This is a strictly cosmetic
                     choice: all of our results can be replayed \emph{mutatis mutandis} in
                     the standard direct-style setting as well.)
                    ANF enforces an order of evaluation and it requires
                    that all arguments to a function be atomic:
                    \begin{align*}
                    \expr \in \syn{Exp} &\produces \letiform{\vv}{\call}{\expr} && \text{[non-tail call]}
                    \\
                      &\;\;\opor\;\; \call && \text{[tail call]}
                      \\
                        &\;\;\opor\;\; \aexpr && \text{[return]}
                        \\
                          \fexpr,\aexpr \in \syn{Atom} &\produces \vv \opor \lam && \text{[atomic expressions]}
                          \\
                            \lam \in \syn{Lam} &\produces \lamform{\vv}{\expr} && \text{[lambda terms]}
                            \\
                              \call \in \syn{Call} &\produces \appform{\fexpr}{\aexpr} && \text{[applications]}
                              \\
                                \vv \in \syn{Var} &\text{ is a set of identifiers} && \text{[variables]}
                                \text.
                                \end{align*}

                                We use the CESK machine of Felleisen and
                                Friedman~\cite{mattmight:Felleisen:1987:CESK} to specify a small-step semantics
                                for ANF.
                                The CESK machine has an explicit stack, and under a structural abstraction, the stack
                                component of this machine directly becomes the stack component of a
                                pushdown system.
                                The set of configurations ($\s{Conf}$) for 
                                this machine has the four expected components (Figure~\ref{fig:cesk}).
                                \begin{figure}
                                \begin{align*}
                                \conf \in \s{Conf} &= \syn{Exp} \times \s{Env} \times \s{Store} \times \s{Kont} && \text{[configurations]}
                                \\
                                  \env \in \s{Env} &= \syn{Var} \parto \s{Addr} && \text{[environments]}
                                  \\
                                    \store \in \s{Store} &= \s{Addr} \to \s{Clo} && \text{[stores]}
                                      \\
                                        \clo \in \s{Clo} &= \syn{Lam} \times \s{Env} && \text{[closures]}
                                        \\
                                          \cont \in \s{Kont} &= \s{Frame}^* && \text{[continuations]}
                                          \\
                                            \phrame \in \s{Frame} &= \syn{Var} \times \syn{Exp} \times \s{Env}  && \text{[stack frames]}
                                            \\
                                              \addr \in \s{Addr} &\text{ is an infinite set of addresses} && \text{[addresses]}
                                              \text.
                                              \end{align*}
                                              \caption{The concrete configuration-space.}
                                              \label{fig:cesk}
                                              \end{figure}

                                              \subsection{Semantics}

                                              To define the semantics, we need five items:
                                              \begin{enumerate}
                                              \item $\Inject : \syn{Exp} \to \s{Conf}$ injects an expression into
                                              a configuration:
                                              \begin{equation*}
\conf_0 = \Inject(\expr) = (\expr, [], [], \vect{})
  \text.
  \end{equation*}

  \item $\ArgEval : \syn{Atom} \times \s{Env} \times \s{Store} \parto
  \s{Clo}$ evaluates atomic expressions:
  \begin{align*}
  \ArgEval(\lam,\env,\store) &= (\lam,\env) && \text{[closure creation]}
  \\
    \ArgEval(\vv,\env,\store) &= \store(\env(\vv)) && \text{[variable look-up]}
    \text.
    \end{align*}

    \item $(\To) \subseteq \s{Conf} \times \s{Conf}$ transitions between
configurations. (Defined below.)

  \item $\Eval : \syn{Exp} \to \Pow{\s{Conf}}$ computes the set of
  reachable machine configurations for a given program:
  \begin{equation*}
  \Eval(\expr) = \setbuild{ \conf }{ \Inject(\expr) \To^* \conf } 
  \text.
  \end{equation*}

  \item $\alloc : \syn{Var} \times \s{Conf} \to \s{Addr}$ 
  chooses fresh store addresses for newly bound variables.
  The address-allocation function is an opaque parameter in this
  semantics,
  so that the forthcoming abstract semantics may also
  parameterize allocation.
  This parameterization provides the knob to tune the
  polyvariance and context-sensitivity of the resulting analysis.
  For the sake of defining the concrete semantics, letting addresses be
  natural numbers suffices, and then the allocator can choose the lowest
  unused address:
  \begin{align*}
  \s{Addr} &= \mathbb{N}
  \\
  \alloc(v,(\expr,\env,\store,\cont)) &= 
1 + \max(\dom(\store))
  \text.
  \end{align*}

  \end{enumerate}




  \paragraph{Transition relation}
  To define the transition $\conf \To \conf'$, we need three rules.
  The first rule handle tail calls by evaluating the function into a
  closure, evaluating the argument into a value and then moving to the
  body of the closure's \lamterm{}:
  \begin{align*}
  \overbrace{(\sembr{\appform{\fexpr}{\aexpr}}, \env, \store, \cont)}^{\conf}
  &\To
  \overbrace{(\expr,\env'',\store',\cont)}^{\conf'}
  \text{, where }
  \\
    (\sembr{\lamform{\vv}{\expr}}, \env') &= \ArgEval(\fexpr,\env,\store)
    \\
      \addr &= \alloc(\vv,\conf)
      \\
        \env'' &= \env'[\vv \mapsto \addr]
        \\
          \store' &= \store[\addr \mapsto \ArgEval(\aexpr,\env,\store)]
          \text.
          \end{align*}

          \noindent
          Non-tail call pushes a frame onto the stack and evaluates the call:
                    \begin{align*}
                    \overbrace{(\sembr{\letiform{\vv}{\call}{\expr}}, \env, \store, \cont)}^{\conf}
                    &\To
                    \overbrace{(\call,\env,\store, (\vv,\expr,\env) : \cont)}^{\conf'}
                    \text.
                    \end{align*}

                    \noindent
                    Function return pops a stack frame:
                    \begin{align*}
                    \overbrace{(\aexpr, \env, \store, (\vv,\expr,\env') : \cont)}^{\conf}
                    &\To
                    \overbrace{(\expr,\env'',\store', \cont)}^{\conf'}
                    \text{, where }
                    \\
                      \addr &= \alloc(\vv,\conf)
                      \\
                        \env'' &= \env'[\vv \mapsto \addr]
                        \\
                          \store' &= \store[\addr \mapsto \ArgEval(\aexpr,\env,\store)]
                          \text.  
                          \end{align*}

                          \section{Pushdown abstract interpretation}
                          \label{sec:abstraction}
                          Our first step toward a static analysis 
                          is an abstract interpretation
                          into an \emph{infinite} state-space.
                          To achieve a pushdown analysis,
                          we simply
                          abstract away less than we normally would.
                          Specifically, 
                          we leave the stack height unbounded.

                          Figure~\ref{fig:abs-conf-space}
                          details the 
                          abstract configuration-space.
                          To synthesize it, we force addresses to be a finite set, but
                          crucially, we leave the stack untouched.
                          When we compact the set of addresses into a finite set, the machine
                          may run out of addresses to allocate, and when it does, the
                          pigeon-hole principle will force multiple closures to reside at the
                          same address.
                          As a result, we have no choice but to force the range of the store to
                          become a power set in the abstract configuration-space.
                          The abstract transition relation has  components
                          analogous to those from the concrete semantics:

                          \begin{figure}
                          \begin{align*}
                          \aconf \in \sa{Conf} &= \syn{Exp} \times \sa{Env} \times \sa{Store} \times \sa{Kont} && \text{[configurations]}
                          \\
                            \aenv \in \sa{Env} &= \syn{Var} \parto \sa{Addr} && \text{[environments]}
                            \\
                              \astore \in \sa{Store} &= \sa{Addr} \to \Pow{\sa{Clo}} && \text{[stores]}
                                \\
                                  \aclo \in \sa{Clo} &= \syn{Lam} \times \sa{Env} && \text{[closures]}
                                  \\
                                    \acont \in \sa{Kont} &= \sa{Frame}^* && \text{[continuations]}
                                    \\
                                      \aphrame \in \sa{Frame} &= \syn{Var} \times \syn{Exp} \times \sa{Env}  && \text{[stack frames]}
                                      \\
                                        \aaddr \in \sa{Addr} &\text{ is a \emph{finite} set of addresses} && \text{[addresses]}
                                        \text.
                                        \end{align*}
                                        \caption{The abstract configuration-space.}
                                        \label{fig:abs-conf-space}
                                        \end{figure}

                                        \paragraph{Program injection}
                                        The abstract injection function $\aInject : \syn{Exp} \to \sa{Conf}$
                                        pairs an expression with an empty environment, an empty store and an
                                        empty stack to create the initial abstract configuration:
                                        \begin{equation*}
\aconf_0 = \aInject(\expr) = (\expr, [], [], \vect{})
  \text.
  \end{equation*}

  \paragraph{Atomic expression evaluation}
  The abstract atomic expression evaluator, $\aArgEval : \syn{Atom}
  \times \sa{Env} \times \sa{Store} \to \PowSm{\sa{Clo}}$, returns the value of
  an atomic expression in the context of an environment and a store;
  it returns a \emph{set} of abstract closures:
  \begin{align*}
  \aArgEval(\lam,\aenv,\astore) &= \set{(\lam,\aenv)} && \text{[closure creation]}
  \\
    \aArgEval(\vv,\aenv,\astore) &= \astore(\aenv(\vv)) && \text{[variable look-up]}
    \text.
    \end{align*}

    \paragraph{Reachable configurations}
    The abstract program evaluator $\aEval : \syn{Exp} \to
    \PowSm{\sa{Conf}}$ returns all of the configurations reachable from
    the initial configuration:
    \begin{equation*}
    \aEval(\expr) = \setbuild{ \aconf }{ \aInject(\expr) \aTo^* \aconf } 
    \text.
    \end{equation*}
    Because there are an infinite number of abstract configurations, a
    na\"ive implementation of this function may not terminate.

    \paragraph{Transition relation}
    The abstract transition relation $(\aTo) \subseteq \sa{Conf} \times
    \sa{Conf}$ has three rules, one of which has become nondeterministic.
    A tail call may fork because there could be multiple abstract closures
    that it is invoking:
    \begin{align*}
    \overbrace{(\sembr{\appform{\fexpr}{\aexpr}}, \aenv, \astore, \acont)}^{\aconf}
    &\aTo
    \overbrace{(\expr,\aenv'',\astore',\acont)}^{\aconf'}
    \text{, where }
    \\
      (\sembr{\lamform{\vv}{\expr}}, \aenv') &\in \aArgEval(\fexpr,\aenv,\astore)
      \\
        \aaddr &= \aalloc(\vv,\aconf)
        \\
          \aenv'' &= \aenv'[\vv \mapsto \aaddr]
          \\
            \astore' &= \astore \join [\aaddr \mapsto \aArgEval(\aexpr,\aenv,\astore)]
            \text.
            \end{align*}
            We define all of the partial orders shortly, but for stores:
            \begin{equation*}
            (\astore \join \astore')(\aaddr) = \astore(\aaddr) \union \astore'(\aaddr)
            \text.
            \end{equation*}

            \noindent
            A non-tail call pushes a frame onto the stack and evaluates the call:
            \begin{align*}
            \overbrace{(\sembr{\letiform{\vv}{\call}{\expr}}, \aenv, \astore, \acont)}^{\aconf}
            &\aTo
            \overbrace{(\call,\aenv,\astore, (\vv,\expr,\aenv) : \acont)}^{\aconf'} 
            \text.
            \end{align*}

            \noindent
            A function return pops a stack frame:
            \begin{align*}
            \overbrace{(\aexpr, \aenv, \astore, (\vv,\expr,\aenv') : \acont)}^{\aconf}
            &\aTo
            \overbrace{(\expr,\aenv'',\astore', \acont)}^{\aconf'}
            \text{, where }
            \\
              \aaddr &= \aalloc(\vv,\aconf)
              \\
                \aenv'' &= \aenv'[\vv \mapsto \aaddr]
                \\
                  \astore' &= \astore \join [\aaddr \mapsto \aArgEval(\aexpr,\aenv,\astore)]
                  \text.  
                  \end{align*}

                  \paragraph{Allocation: Polyvariance and context-sensitivity}
                  \label{sec:polyvariance}
                  In the abstract semantics, the abstract allocation function
                  $\aalloc : \syn{Var} \times \sa{Conf} \to \sa{Addr}$ determines the
                  polyvariance of the analysis.
                  In a control-flow analysis, \emph{polyvariance} literally refers to
                  the number of abstract addresses (variants) there are for each
                  variable.
                  An advantage of this framework over CFA2 is
                  that varying this abstract allocation function
                  instantiates pushdown versions of classical flow analyses.
                  All of the following allocation approaches can be used with the
                  abstract semantics.
                  The abstract allocation function is a
                  parameter to the analysis.

                  \paragraph{Monovariance: Pushdown 0CFA}

                  Pushdown 0CFA uses variables themselves for abstract addresses:

                  \begin{align*}
                  \sa{Addr} &= \syn{Var}
                  \\
                    \alloc(v,\aconf) &= v
                    \text.
                    \end{align*}

                    \paragraph{Context-sensitive: Pushdown 1CFA}

                    Pushdown 1CFA pairs the variable with the current expression to get
                    an abstract address:

                    \begin{align*}
                    \sa{Addr} &= \syn{Var} \times \syn{Exp}
                    \\
                      \alloc(\vv,(\expr,\aenv,\astore,\acont)) &= (\vv,\expr)
                      \text.
                      \end{align*}

                      \paragraph{Polymorphic splitting: Pushdown poly/CFA}

                      Assuming we compiled the program from a programming language with
                      let-bound polymorphism and marked which functions were let-bound, we
                      can enable polymorphic splitting:

                      \begin{align*}
                      \sa{Addr} &= \syn{Var} + \syn{Var} \times \syn{Exp} 
                      \\
                        \alloc(\vv,(\sembr{\appform{\fexpr}{\aexpr}},\aenv,\astore,\acont)) &= 
                        \begin{cases}
                        (\vv,\sembr{\appform{\fexpr}{\aexpr}}) & \fexpr \text{ is let-bound}
                        \\
                          \vv & \text{otherwise}
                          \text.
                          \end{cases}
                          \end{align*}

                          \paragraph{Pushdown $k$-CFA}

                          For pushdown $k$-CFA, we need to look beyond the current state
                          and at the last $k$ states.
                          By concatenating the expressions in the last $k$ states together, and
                          pairing this sequence with a variable we get pushdown $k$-CFA:
                          \begin{align*}
                          \sa{Addr} &= \syn{Var} \times \syn{Exp}^k
                          \\
                            \aalloc(\vv,\vect{(\expr_1,\aenv_1,\astore_1,\acont_1),\ldots}) &=
(\vv,\vect{\expr_1,\ldots,\expr_k})
  \text.
  \end{align*}

  \subsection{Partial orders}

  For each set $\hat X$ inside the abstract configuration-space, we use the
  natural partial order, $(\wt_{\hat X}) \subseteq \hat X \times \hat X$.
  Abstract addresses and syntactic sets have flat partial orders.
  For the other sets, the
  partial order lifts:
  \begin{itemize}
  \item point-wise over environments:
  \begin{equation*}
  \aenv \wt \aenv' \text{ iff }
  \aenv(\vv) = \aenv'(\vv) \text{ for all } \vv \in \dom(\aenv)
  \text;
  \end{equation*}

  \item component-wise over closures:
  \begin{equation*}
  (\lam,\aenv) \wt (\lam,\aenv') \text{ iff } \aenv \wt \aenv'
  \text;
  \end{equation*}

  \item point-wise over stores:
  \begin{equation*}
  \astore \wt \astore' 
  \text{ iff }
  \astore(\aaddr) \wt \astore'(\aaddr) \text{ for all } \aaddr
\in \dom(\astore)
  \text;
  \end{equation*}

  \item component-wise over frames:
  \begin{equation*}
  (\vv,\expr,\aenv) \wt (\vv,\expr,\aenv') \text{ iff } \aenv \wt \aenv'
  \text;
  \end{equation*}

  \item element-wise over continuations:
  \begin{equation*}
  \vect{\aphrame_1,\ldots,\aphrame_n} 
  \wt 
  \vect{\aphrame'_1,\ldots,\aphrame'_n} 
  \text{ iff }
  \aphrame_i \wt \aphrame'_i
  \text{; and }
  \end{equation*}

  \item component-wise across configurations:
  \begin{equation*}
(\expr,\aenv,\astore,\acont) 
  \wt
  (\expr,\aenv',\astore',\acont')
  \text{ iff }
  \aenv \wt \aenv'
  \text{ and }
  \astore \wt \astore'
  \text{ and }
  \acont \wt \acont'
  \text.
  \end{equation*}
  \end{itemize}


  \subsection{Soundness}

  To prove soundness,  an abstraction map $\absmap$ connects the
  concrete and abstract configuration-spaces:
  \begin{align*}
\absmap(\expr,\env,\store,\cont) &= (\expr,\absmap(\env),\absmap(\store),\absmap(\cont))
  \\
    \absmap(\env) &= \lt{\vv}{\absmap(\env(\vv))}
    \\
      \absmap(\store) &= \lt{\aaddr}{\!\!\! \bigjoin_{\absmap(\addr) = \aaddr} \!\!\! \set{\absmap(\store(\addr))}}
      \\
        \absmap\vect{\phrame_1,\ldots,\phrame_n} &= \vect{\absmap(\phrame_1), \ldots, \absmap(\phrame_n)}
        \\
          \absmap(\vv,\expr,\env) &= (\vv,\expr,\absmap(\env))
          \\
            \absmap(\addr) &\text{ is determined by the allocation functions}
            \text.
            \end{align*}
            It is then easy to prove that the abstract transition relation
            simulates the concrete transition relation:
            \begin{theorem}
If:
\begin{equation*}
\absmap(\conf) \wt \aconf \text{ and } \conf \To \conf'
\text,
  \end{equation*}
  then there must exist $\aconf' \in \sa{Conf}$ such that:
  \begin{equation*}
  \absmap(\conf') \wt \aconf' \text{ and } \aconf \aTo \aconf'
  \text.
  \end{equation*}  
  \end{theorem}
  \begin{proof}
  The proof follows by case-wise analysis on the type of the
  expression in the configuration.
  It is a straightforward adaptation of similar proofs, such as that of
  \cite{mattmight:Might:2007:Dissertation} for $k$-CFA.
  \end{proof}

  \section{The shift: From abstract CESK
    to rooted PDS}
    \label{sec:pda}
    In the previous section, we constructed an infinite-state abstract
    interpretation of the CESK machine.
    The infinite-state nature of the abstraction makes it difficult to see how
    to answer static analysis questions.
    Consider, for instance, a control flow-question:
    \begin{quote}
    At the call site $\appform{\fexpr}{\aexpr}$, may a closure over
    $\lam$ be called?
    \end{quote}
    If the abstracted CESK machine were a finite-state machine, an
    algorithm could answer this question by enumerating all reachable
    configurations and looking for an abstract configuration
    $(\sembr{\appform{\fexpr}{\aexpr}},\aenv,\astore,\acont)$ in which
    $(\lam,\_) \in \aArgEval(\fexpr,\aenv,\astore)$.
    However, because the abstracted CESK machine may contain an infinite
    number of reachable configurations, enumeration is not an option.

    Fortunately, a shift in perspective reveals the abstracted CESK machine to be a
    rooted pushdown system.
    This shift permits the use of a control-state reachability
    algorithm in place of exhaustive search of the
    configuration-space.
    In this shift, a control-state is an expression-environment-store triple, and a
    stack character is a frame.
    Figure~\ref{fig:acesk-to-pds} defines the program-to-RPDS conversion function
    $\afPDS : \syn{Exp} \to \mathbb{RPDS}$.  

    \begin{figure}
    \begin{align*}
\afPDS(\expr) &= (\QStates,\StackAlpha,\transfunction,\qstate_0)
  \text{, where }
  \\
    \QStates &= \syn{Exp} \times \sa{Env} \times \sa{Store}
    \\
      \StackAlpha &= \sa{Frame}
      \\
        (\qstate,\epsilon,\qstate') \in \transfunction
        & \text{ iff }
(\qstate, \acont)
  \aTo
  (\qstate', \acont)
  \text{ for all } \acont
  \\
    (\qstate,\aphrame_{-},\qstate') \in \transfunction
    & \text{ iff }
(\qstate, \aphrame : \acont)
  \aTo
  (\qstate',\acont)
  \text{ for all } \acont
  \\
    (\qstate,\aphrame'_{+},\qstate') \in \transfunction
    & \text{ iff }
(\qstate, \acont)
  \aTo
  (\qstate',\aphrame' : \acont)
  \text{ for all } \acont
  \\
    (\qstate_0,\vect{}) &= \aInject(\expr)
    \text.
    \end{align*}
    \caption{$\afPDS : \syn{Exp} \to \mathbb{RPDS}$.
    }
\label{fig:acesk-to-pds}
\end{figure}

At this point, we can compute the root-reachable control states using a
straightforward summarization algorithm~\cite{mattmight:Bouajjani:1997:PDA-Reachability,mattmight:Reps:1998:CFL,mattmight:Reps:2005:Weighted-PDA}.
This is the essence of CFA2.

\section{Introspection for abstract garbage collection}
Abstract garbage collection~\cite{mattmight:Might:2006:GammaCFA} yields large
improvements in precision by using the abstract interpretation of garbage
collection to make more efficient use of the finite address space available
during analysis.
Because of the way abstract garbage collection operates, it grants exact
precision to the flow analysis of variables whose bindings die
between invocations of the same abstract context.
Because pushdown analysis grants exact precision in tracking return-flow, it is
clearly advantageous to combine these techniques.
Unfortunately, as we shall demonstrate, abstract garbage collection
breaks the pushdown model by requiring full stack inspection to discover the
root set.

Abstract garbage collection modifies the transition relation
to conduct a ``stop-and-copy'' garbage collection before each
transition.
To do this, we define a garbage collection function 
$\aCollect : \sa{Conf} \to \sa{Conf}$
on
configurations:
\begin{align*}
\aCollect(\overbrace{\expr,\aenv,\astore,\acont}^{\aconf})
&= (\expr,\aenv,\astore|\mathit{Reachable}(\aconf),\acont)
  \text,
  \end{align*}
  where the pipe operation $f|S$ yields the function $f$, but with
  inputs not in the set $S$ mapped to bottom---the empty set.
  The reachability function $\mathit{Reachable} : \sa{Conf} \to \PowSm{\sa{Addr}}$
  first computes the root set, and then the transitive closure of an
  address-to-address adjacency relation: 
  \begin{align*}
  \mathit{Reachable}(\overbrace{\expr,\aenv,\astore,\acont}^\aconf) &=
\setbuild{ \aaddr }{ \aaddr_0 \in \mathit{Root}(\aconf)
  \text{ and }
  \aaddr_0  
    \mathrel{\areaches_\astore^*}
  \aaddr
}
\text,
  \end{align*}
  where the function $\mathit{Root} : \sa{Conf} \to 
  \PowSm{\sa{Addr}}$ 
  finds the root addresses:
  \begin{align*} 
  \mathit{Root}(\expr,\aenv,\astore,\acont) &=
  \mathit{range}(\aenv) \union
\mathit{StackRoot}(\acont)
  \text,
  \end{align*}
  and the $\mathit{StackRoot} : \sa{Kont} \to \PowSm{\sa{Addr}}$ function
  finds roots down the stack:
  \begin{align*} 
  \mathit{StackRoot}
  \vect{(\vv_1,\expr_1,\aenv_1),\ldots,
    (\vv_n,\expr_n,\aenv_n)}
  &= 
\Union_i \range(\aenv_i)
  \text,
  \end{align*}
  and the relation
  $(\areaches) \subseteq \sa{Addr} \times \sa{Store} \times \sa{Addr}$
  connects adjacent addresses:
  \begin{align*}
  \aaddr 
  \mathrel{\areaches_\astore} 
  \aaddr'
  \text{ iff there exists }
(\lam,\aenv) \in \astore(\aaddr)
  \text{ such that }
  \aaddr' \in \range(\aenv)
  \text.
  \end{align*}

  The new abstract transition relation is thus the composition of abstract garbage collection with the old transition relation:
  \begin{align*}
  (\aTo_{\mathrm{GC}}) = 
  (\aTo) \compose \aCollect
  \end{align*}

  \paragraph{Problem: Stack traversal violates pushdown constraint}

  In the formulation of pushdown systems, the transition relation is restricted
  to looking at the top frame, and even in less restricted formulations,
  at most a bounded number of frames can be inspected.
  Thus, the relation $(\aTo_{\mathrm{GC}})$ cannot be computed as
  a straightforward pushdown analysis using summarization.

  \paragraph{Solution: Introspective pushdown systems}
  To accomodate the richer structure of the relation $(\aTo_{\mathrm{GC}})$, we
  now define \emph{introspective} pushdown systems.
  Once defined, we can embed the garbage-collecting abstract interpretation
  within this framework, and then focus on developing a control-state
  reachability algorithm for these systems.

  An \defterm{introspective pushdown system} is a quadruple
  $M = (\ControlStates,\StackAlpha,\transfunction,\qstate_0)$:
  \begin{enumerate}

  \item $\ControlStates$ is a finite set of control states;

  \item $\StackAlpha$ is a stack alphabet; 

  \item $\transfunction \subseteq \ControlStates \times \StackAlpha^*
  \times \StackAlpha_\pm \times \ControlStates$ is a transition relation; and

  \item $\qstate_0$ is a distinguished root control state. 
  \end{enumerate}
  The second component in the transition relation is 
  a realizable stack at the given control-state.
  This realizable stack distinguishes an introspective pushdown system
  from a general pushdown system.
  $\mathbb{IPDS}$  denotes the class of all introspective pushdown
  systems.

  Determining how (or if) a control state $\qstate$
  transitions to a control state $\qstate'$, requires knowing a
  path taken to the state $\qstate$.
  Thus, we need to define reachability inductively.
  When $M = (\ControlStates,\StackAlpha,\transfunction,\qstate_0)$,
  transition from the initial control state considers only
  empty stacks:
  \begin{align*}
  \qstate_0  
  \mathrel{\RPDTrans_M^\stackact}
  \qstate
  \text{ iff }
(\qstate_0, \vect{}, \stackact, \qstate)
  \in \transfunction
  \text.
  \end{align*}
  For non-root states, the paths to that state matter,
  since they determine the stacks realizable with that state:
  \begin{align*}
  \qstate 
  \mathrel{\RPDTrans_M^\stackact}
  \qstate' 
  &\text { iff there exists }
  \vec{\stackact}
  \text{ such that }
  \qstate_0
  \mathrel{\RPDTrans_M^{\vec{\stackact}}}
  \qstate
  \text{ and }
  (\qstate, [\vec{\stackact}], \stackact, \qstate' )
  \in
  \transfunction,
  \\
    &  \text{ where }\qstate
    \mathrel{\RPDTrans_M^{\vect{\stackact_1,\ldots,\stackact_n}}}
    \qstate'
    \text{ iff }
    \qstate
    \mathrel{\RPDTrans_M^{\stackact_1}}
    \qstate_1
    \mathrel{\RPDTrans_M^{\stackact_2}}
    \cdots
    \mathrel{\RPDTrans_M^{\stackact_n}}
    \qstate'
    \text.
    \end{align*}

    \subsection{Garbage collection in introspective pushdown systems}

    To convert the garbage-collecting,
    abstracted CESK machine into an introspective pushdown system,
    we use the function $\afIPDS : \syn{Exp} \to \mathbb{IPDS}$:
    \begin{align*}
\afIPDS(\expr) &= (\QStates,\StackAlpha,\transfunction,\qstate_0)
  \\
    \QStates &= \syn{Exp} \times \sa{Env} \times \sa{Store}
    \\
      \StackAlpha &= \sa{Frame}
      \\
        (\qstate,\acont,\epsilon,\qstate') 
        \in \transfunction
        & \text{ iff }
\aCollect(\qstate, \acont)
  \aTo
  (\qstate', \acont)
  \\
    (\qstate,\aphrame : \acont,\aphrame_{-},\qstate')
    \in \transfunction
    & \text{ iff }
\aCollect(\qstate, \aphrame : \acont) 
  \aTo
  (\qstate',\acont)
  \\
    (\qstate,\acont,\aphrame_{+},\qstate') 
    \in \transfunction
    & \text{ iff }
\aCollect(\qstate, \acont)
  \aTo
  (\qstate',\aphrame : \acont)
  \\
    (\qstate_0,\vect{}) &= \aInject(\expr)
    \text.
    \end{align*}

    \section{Introspective reachability via Dyck state graphs}
    \label{sec:pdreachability}

    Having defined introspective pushdown systems and embedded
    our abstract, garbage-collecting semantics within them, we
    are ready to define control-state reachability for IDPSs.

    We cast our reachability algorithm for
    introspective pushdown systems 
    as
    finding a fixed-point, in
    which we incrementally accrete the reachable control states
    into a ``Dyck state graph.''

    A \defterm{Dyck state graph} is a quadruple
    $G = (\DSStates,\DSFrames,\DSEdges,\dsstate_0)$, in which:
    \begin{enumerate}
    \item $\DSStates$ is a finite set of nodes;
    \item $\DSFrames$ is a set of frames;
    \item $\DSEdges \subseteq \DSStates \times \DSFrames_\pm \times
    \DSStates$ is a set of stack-action edges; and
    \item $\dsstate_0$ is an initial state;
    \end{enumerate}
    such that for any node $\dsstate \in \DSStates$, it must be the case that:
    \begin{equation*}
(\dsstate_0,\vect{}) \mathrel{\PDTrans_G^*} (\dsstate,\vec{\dsframe})
  \text{ for some stack } \vec{\dsframe}
  \text.
  \end{equation*}
  In other words, a Dyck state graph is equivalent to a rooted
  pushdown system in which there is a legal path to every control state
  from the initial control state.\footnote{We chose the term \emph{Dyck
    state graph} because the sequences of stack actions along valid paths
      through the graph correspond to substrings in Dyck languages.
      A \defterm{Dyck language} is a language of balanced, ``colored''
      parentheses.
      In this case, each character in the stack alphabet is a color.}
      We use $\mathbb{DSG}$ to denote the class of Dyck state graphs.
(Clearly, $\mathbb{DSG} \subset \mathbb{RPDS}$.)

  Our goal is to compile an 
  implicitly-defined introspective pushdown system into 
  an explicited-constructed Dyck state graph.
  During this transformation, the per-state path
  considerations
  of an introspective pushdown are ``baked into''
  the Dyck state graph.
  We can formalize this compilation process as a map, 
  $\fDSG : \mathbb{IPDS} \to
  \mathbb{DSG}\text.$

  Given an introspective pushdown system $M =
  (\ControlStates,\StackAlpha,\transfunction,\qstate_0)$, its equivalent Dyck state
  graph is $\fDSG(M) = (\DSStates,\DSFrames,\DSEdges,\qstate_0)$, where 
  $\dsstate_0 = \qstate_0$,
  the set $\DSStates$ contains reachable nodes:
  \begin{equation*}
  \DSStates = 
  \setbuild{ \qstate }{ 
      \qstate_0
      \mathrel{\RPDTrans_M^{\vec{\stackact}}}
      \qstate
      \text{ for some stack-action sequence }
      \vec{\stackact}
      }
\text,
  \end{equation*}
  and the set $\DSEdges$ contains reachable edges:
  \begin{equation*}
  \DSEdges = \setbuild{ \qstate \pdedge^{\stackact} \qstate' }{
      \qstate
      \mathrel{\RPDTrans_M^{\stackact}}
    \qstate'
      %
  }
\text.
\end{equation*}
Our goal is to find a method
for computing a Dyck state graph from 
an introspective pushdown
system.

%
%

\subsection{Compiling to Dyck state graphs}
\label{sec:ipds-to-dsg}
We now turn our attention to compiling an introspective pushdown
system (defined implicitly) into a Dyck state graph (defined
    explicitly).
That is, we want an implementation of the function $\fDSG$.
To do so, we first phrase the Dyck state graph construction as the
least fixed point of a monotonic function.
This formulation provides a straightforward iterative method for computing
the function $\fDSG$.
%

The function $\mkDSG : \mathbb{IPDS} \to (\mathbb{DSG} \to
    \mathbb{DSG})$ generates the monotonic iteration function we need:
\begin{align*}
\mkDSG(M) &= f\text{, where }
\\
    M &= (\QStates,\StackAlpha,\transfunction,\qstate_0)
    \\
      f(\DSStates,\DSFrames,\DSEdges,\dsstate_0) &= (\DSStates',\DSFrames,\DSEdges',\dsstate_0) \text{, where }
      \\
        \DSStates' &= \DSStates \union \setbuild{ \dsstate' }{ 
          \dsstate \in \DSStates 
            \text{ and }  
          \dsstate 
            \mathrel{\RPDTrans_M}
          \dsstate'
        }
\union \set{\dsstate_0}
\\
    \DSEdges' &= \DSEdges \union \setbuild{ \dsstate \pdedge^\stackact \dsstate' }{ 
      \dsstate \in \DSStates 
        \text{ and }  
      \dsstate 
        \mathrel{\RPDTrans_M^\stackact}
      \dsstate'    
    }
\text.
\end{align*}
Given an introspective pushdown system $M$, each application of the function
$\mkDSG(M)$ accretes new edges at the frontier of the Dyck state
graph.

\subsection{Computing a round of $\mathcal{F}$}
\label{sec:stacks-nfa}
The formalism obscures an important detail in the
computation of an iteration: 
the transition relation $(\RPDTrans)$ for the introspective
pushdown system must compute all possible stacks
in determining whether or not there exists a transition.
Fortunately, this is not as onerous as it seems: the set of all possible stacks 
for any given control-point is a regular language,
    and the finite automaton that encodes this language
    can be lifted (or read off) the structure of the Dyck state graph.
    The function $\mathit{Stacks} : \mathbb{DSG} \to \DSStates \to \mathbb{NFA}$
    performs exactly this extraction:
    \begin{align*}
    \mathit{Stacks}
    (
     \overbrace{\DSStates, \StackAlpha, \DSEdges, \dsstate_0}^{M}
    )
(\dsstate)
  &= 
  (\DSStates, \StackAlpha, \transfunction, \dsstate_0, 
   \set{\dsstate} )\text{, where }
  \\
    (\dsstate',\stackchar,\dsstate'')
    \in \transfunction
    & \text{ if }
    (\dsstate', \stackchar_+, \dsstate'') \in \DSEdges
    \\
      (\dsstate',\epsilon,\dsstate'')
      \in \transfunction
      & \text{ if }
      \dsstate' \RPDTrans_M^{\vec{\stackact}}
      \dsstate''
      \text{ and }
      [\vec{\stackact}] = \epsilon
      \text.
      \end{align*}


      \subsection{Correctness}
      Once the algorithm reaches a fixed point, the Dyck state graph is complete:
      \begin{theorem} 
      $\fDSG(M) = \lfp(\mkDSG(M))$.
      \end{theorem}
      \begin{proof}
      Let $M = (\QStates,\StackAlpha,\transfunction,\qstate_0)$.
      Let $f = \mkDSG(M)$.
      Observe that $\lfp(f) = f^n(\emptyset,\StackAlpha,\emptyset,\qstate_0)$ for some $n$.
      When $N \subseteq M$, then it easy to show that $f(N) \subseteq M$.
      Hence, $\fDSG(M) \supseteq \lfp(\mkDSG(M))$.

      To show $\fDSG(M) \subseteq \lfp(\mkDSG(M))$, suppose this is not
      the case.
      Then, there must be at least one edge in $\fDSG(M)$ that is not in
      $\lfp(\mkDSG(M))$.
      By the defintion of $\fDSG(M)$, each edge must be part of a sequence of edges from the initial
      state.
      Let $(\dsstate,\stackact,\dsstate')$ be the first edge in its sequence from the initial state
      that is not in $\lfp(\mkDSG(M)$.
          Because the proceeding edge is in $\lfp(\mkDSG(M))$, the state $\dsstate$ \emph{is} in 
          $\lfp(\mkDSG(M))$.
          Let $m$ be the lowest natural number such that $\dsstate$ appears in $f^m(M)$.
          By the definition of $f$, this edge must appear in $f^{m+1}(M)$, which means it must also appear in 
          $\lfp(\mkDSG(M))$, which is a contradiction.
          Hence, $\fDSG(M) \subseteq \lfp(\mkDSG(M))$.
          \end{proof}

          \subsection{Complexity}

          While decidability is the goal, it is straightforward
          to determine the complexity of this na\"ive fixed-point
          method.
      To determine the complexity of this algorithm, we ask two questions:
      how many times would the algorithm invoke the iteration function in
      the worst case, and how much does each invocation cost in the
      worst case?
      The size of the final Dyck state
      graph bounds the run-time of the algorithm.
      Suppose the final Dyck state graph has $m$ states.
      In the worst case, the iteration function adds only a single edge each
      time.
      Between any two states, there is one $\epsilon$-edge, one push edge, or some
      number of pop edges (at most $\abs{\StackAlpha}$).
      Since there are at most $\abs{\StackAlpha}m^2$ edges in the final graph, the maximum
      number of iterations is $\abs{\StackAlpha}m^2$.

      The cost of computing each iteration is harder to bound.
      The cost of determining whether to add a push edge is constant, as is
      the cost of adding an $\epsilon$-edge.
      So the cost of determining all new push edges and new $\epsilon$-edges
      to add is constant.
      %
      Determining whether or not to add a pop edge is expensive.
      To add the pop edge
      $\triedge{\dsstate}{\stackchar_-}{\dsstate'}$, we must prove that
      there exists a configuration-path to the control state $\dsstate$, in which
      the character $\stackchar$ is on the top of the stack.
      This reduces to a CFL-reachability query~\cite{mattmight:Melski:2000:CFL}
      at each node, the cost of which is $O(\abs{\StackAlpha_\pm}^3
          m^3)$~\cite{mattmight:Kodumal:2004:CFL}.

      To summarize, in terms of the number of reachable control states, the
      complexity of this naive algorithm is:
  \[
O((\abs{\StackAlpha}m^2) 
    \times
    (\abs{\StackAlpha_\pm}^3m^3)) 
  =  O(\abs{\StackAlpha}^4m^5)\text.
  \]
  (As with summarization, it is possible
   to maintain a work-list and 
   introduce an 
   \ecg{} to avoid spurious recomputation.
   %
   %
   This ultimately reduces complexity to
   $O(\abs{\StackAlpha}^2m^4)$.)


  %

  \section{Implementation and evaluation}

  We have developed an implementation to produce the Dyck state graph of an
  introspective pushdown system. 
  While the fixed-point computation~\ref{sec:stacks-nfa} could be rendered
  directly as functional code, extending the classical summarization-based algorithm 
  for pushdown reachability to introspective pushdown systems yields better performance.
  In this section we present a variant of such an algorithm and discuss results
  from an implementation that can analyze a large subset of the Scheme programming
  language.

  \subsection{Iterating over a DSG: An implementor's view}
  \label{sec:epsilon-algo}


  To synthesize a Dyck state graph from an introspective pushdown system, 
  it is built incrementally---node by node, edge by edge.
  The na\"ive fixed point algorithm presented earlier, if implemented literally,
  would (in the worst case) have to re-examine the entire DSG to add each edge.
  To avoid such re-examination, our implementation adds $\epsilon$-summary edges
  to the DSG.

  In short, an $\epsilon$-summary edge connects two control states if there
  exists a path between them with no net stack change---that is, all pushes are
  cancelled by corresponding pops.
  With $\epsilon$-summary edges available, any change to the graph can be
  propagated directly to where it has an effect, and then any new $\epsilon$-summary edges
  that propagation implies are added.


  Whereas the correspondence between CESK and an IPDS is
  relatively straightforward, the relationship between a DSG and its original IPDS
  is complicated by the fact that the IPDS keeps track of the
  \emph{whole} stack, whereas the DSG distributes (the same) stack information
  throughout its internal structure. 

  A classic reachability-based analysis for a pushdown system requires two
  mutually-dependent pieces of information in order to add another edge:

  \begin{enumerate}
  \item The topmost frame on a stack for a given control state $q$. This is
  essential for \emph{return} transitions, as this frame should be
  popped from the stack and the store and the environment of a caller
  should be updated respectively.

  \item Whether a given control state $q$ is reachable or not from the
  initial state $q_0$ along realizable sequences of stack actions.
  For example, a path from $q_0$ to $q$ along edges labeled ``push, pop, pop, push''
  is not realizable: the stack is empty after the first pop, so the second pop
  cannot happen---let alone the subsequent push.

  %

  \end{enumerate}

  These two data are enough for a classic pushdown reachability summarization
  to proceed
  one step further. 
  However, the presence of an abstract garbage
  collector, and the graduation to an \emph{introspective} pushdown system, imposes the requirement for a third item of data:

  \begin{enumerate}
  \item[3.] For a given control state $q$, what are \emph{all} possible frames
  that could happen to be \emph{on} the stack  at the moment the IPDS
  is in the state $q$?
  \end{enumerate}

  It is possible to recompute these frames from scratch in each iteration using the
NFA-extraction technique we described. 
But, it is easier to maintain per-node summaries, in the same spirit as
$\epsilon$-summary edges.


A version of the classic pushdown summarization algorithm that maintains 
the first two items is presented in~\cite{mattmight:Earl:2010:Pushdown}, 
so we will just outline the key differences here.

The crux of the algorithm is to maintain for each node $\qstate'$
in the DSG, a set of $\epsilon$-\emph{predecessors},
i.e., nodes $\qstate$, such that $\qstate
\RPDTrans_M^{\vec{\stackact}} \qstate'$ and
$[\vec{\stackact}] = \epsilon$. 
In fact, only two out of three kinds
of transitions can cause a change to the set of
$\epsilon$-predecessors for a particular node $\qstate$: an addition of
an $\epsilon$-edge or a pop edge to the DSG. 

It is easy to see why the second action might introduce new
$\epsilon$-paths  and, therefore, new $\epsilon$-predecessors.
Consider, for example, adding the ${\stackchar_-}$-edge $\qstate
\pdedge^{{\stackchar_-}} \qstate'$ into the following graph:
\begin{equation*}
  \xymatrix{ 
    \qstate_0 \ar[r]^{\stackchar_+} & \qstate & \qstate' \ar[r]^{\epsilon} & \qstate_1
  } 
\end{equation*}
As soon this edge drops in, there becomes an ``implicit''
$\epsilon$-edge between $\qstate_0$ and $\qstate_1$ because the net
stack change between them is empty; the resulting graph looks like:
\begin{equation*}
  \xymatrix{ 
    \qstate_0 \ar[r]^{\stackchar_+} \ar @{-->} @(u,u) [rrr]^\epsilon &
    \qstate \ar[r]^{\stackchar_-} &
    \qstate' \ar[r]^{\epsilon} &
    \qstate_1
  }  
\end{equation*}
where we have illustrated the implicit $\epsilon$-edge as a dashed
line.

A little reflection on $\epsilon$-predecessors
and top frames reveals a mutual dependency between
these items during the construction of a DSG. 
Informally:

\begin{itemize}
\item A \emph{top frame} for a state $q$ can be 
  pushed as a direct predecessor, or as
  a direct predecessor to an $\epsilon$-predecessor.

\item When a new $\epsilon$-edge $\qstate \xrightarrow{\epsilon}
  \qstate'$ is added, all $\epsilon$-predecessors of $\qstate$ become
  also $\epsilon$-predecessors of $\qstate'$. 
  That is, $\epsilon$-summary edges are transitive.

\item When a $\stackchar_-$-pop-edge $\qstate
  \xrightarrow{\stackchar_-} \qstate'$ is added, new
  $\epsilon$-predecessors of a state $\qstate_1$ can be obtained by
  checking if $\qstate'$ is an $\epsilon$-predecessor of $\qstate_1$
  and examining all existing $\epsilon$-predecessors of $\qstate$,
  such that $\stackchar_+$ is their possible top frame: this situation
  is similar to the one depicted in the example above. 

\end{itemize}

The third component---\emph{all} possible frames on the stack for a state
$\qstate$---is straightforward to compute with $\epsilon$-predecessors:
starting from $\qstate$, trace out only the edges which are labeled $\epsilon$
(summary or otherwise) or $\stackchar_+$.
The frame for any action $\stackchar_+$ in this trace is a possible stack
action.
Since these sets grow monotonically, it is easy to cache the results
of the trace, and in fact, propagate incremental changes to these
caches when new $\epsilon$-summary or $\stackchar_+$ nodes are
introduced. 
%
Our implementation directly reflects the optimizations discussed above.

\begin{figure*}
\centering

\renewcommand{\goodsingl}[1]{#1}

\begin{tabular}{|l|c|c|c||c|c|c||c|c|c||c|c|c||c|c|c|}
\hline
Program & 
$\syn{Exp}$ & $\syn{Var}$ & $k$ &
\multicolumn{3}{c||}{$k$-CFA} &
\multicolumn{3}{c||}{$k$-PDCFA} &
\multicolumn{3}{c||}{$k$-CFA $+$ GC} &
\multicolumn{3}{c|}{$k$-PDCFA $+$ GC} \\

\hline \hline


\multirow{2}{*}{\texttt{mj09}} & 
\multirow{2}{*}{19} & 
\multirow{2}{*}{8} &
0 &
83 & 107 & 4  &
38 & 38 & 4 &
36 & 39 & 4  &
33 & 32 & 4 \\
\cline{4-16}
&&&
1 &
454 & 812 & 1 &
44 & 48 & 1 &
34 & 35 & 1 &
32 & 31 & 1 \\
\hline


\multirow{2}{*}{\texttt{eta}} & 
\multirow{2}{*}{21} & 
\multirow{2}{*}{13} &
0 &
63 & 74 & 4 &
34 & 34 & 6 &
28 & 27 & 8 &
28 & 27 & 8 \\
\cline{4-16}
&&&
1 &
33 & 33 & 8 &
32 & 31 & 8 &
28 & 27 & 8 &
28 & 27 & 8 \\
\hline


\multirow{2}{*}{\texttt{kcfa2}} & 
\multirow{2}{*}{20} & 
\multirow{2}{*}{10} &
0 &
194 & 236  & 3 &
36 & 35 & 4 &
35 & 43 & 4 &
35 & 34 & 4 \\
\cline{4-16}
&&&
1 &
970 & 1935 & 1 &
87 & 144 & 2 &
35 & 34  & 2 &
35 & 34  & 2  \\
\hline


\multirow{2}{*}{\texttt{kcfa3}} & 
\multirow{2}{*}{25} & 
\multirow{2}{*}{13} &
0 &
272 & 327 & 4 &
58 & 63 & 5 &
53 & 52 & 5 &
53 & 52 & 5 \\
\cline{4-16}
&&&
1 &
$>$ 7119 & $>$  14201 & $\le$ 1 &
1761 & 4046 & 2  &
53 & 52 & 2 &
53 & 52 & 2 \\
\hline


\multirow{2}{*}{\texttt{blur}} & 
\multirow{2}{*}{40} & 
\multirow{2}{*}{20} &
0 &
$>$  1419 & $>$ 2435 & $\le$ 3  &
280 & 414 & 3 &
274 & 298 & 9 &
164 & 182 & 9 \\
\cline{4-16}
&&&
1 &
261 & 340 & 9 &
177 & 189 & 9 &
169 & 189 & 9 &
167 & 182 & 9 \\
\hline


\multirow{2}{*}{\texttt{loop2}} & 
\multirow{2}{*}{41} & 
\multirow{2}{*}{14} &
0 &
228 & 252 & 4 &
113 & 122 & 4 &
86 & 93 & 4 &
70 & 74  & 4  \\
\cline{4-16}
&&&
1 &
$>$ 10867 & $>$ 16040 & $\le$ 3 &
411 & 525 & 3  &
151 & 163 & 3  &
145 & 156 & 3 \\
\hline


\multirow{2}{*}{\texttt{sat}} & 
\multirow{2}{*}{63} & 
\multirow{2}{*}{31} &
0 &
$>$ 5362 & $>$ 7610 & $\le$ 6 &
775 & 979 & 6 &
1190 & 1567 & 6 &
321 & 384 & 6 \\
\cline{4-16}
&&&
1 &
$>$ 8395 & $>$ 12391 & $\le$ 6 &
7979 & 10299 & 6 &
982 & 1330 & 7 &
107 & 106 & 13 \\
\hline


\end{tabular}

\caption{Benchmark results. The first three columns provide the name of a
  benchmark, the number of expressions and variables in the program in
  the ANF, respectively. For each of eight combinations of
  pushdown analysis, $k \in \set{0, 1}$ and garbage collection
  on or off, the first two columns in a group show the number of
  \emph{control states} and transitions/DSG edges computed during the
  analysis (for both less is better). The third column presents the
  amount of \emph{singleton} variables, i.e, how many variables have a
  single lambda flow to them (more is better). Inequalities for some
  results denote the case when the analysis did not finish within 
  30 minutes. For such cases we can only report an upper bound
  of singleton variables as this number can only decrease. 
  %
  }
\label{fig:table-results}
\end{figure*}

\subsection{Experimental results}
\label{sec:experiments}

A fair comparison between different families of analyses 
should compare both precision and speed.
%
%
We have extended an existing
implementation of $k$-CFA to optionally enable pushdown analysis,
abstract garbage collection or both.  
Our implementation source and benchmarks are available:
\begin{center}
\url{http://github.com/ilyasergey/reachability}
\end{center}

As expected, the fused analysis does at least as well as the best of either
analysis alone in terms of singleton flow sets (a good metric for
program optimizability) and better than both in some cases.
Also worthy of note is the dramatic reduction in the size of the
abstract transition graph for the fused analysis---even on top of the already
large reductions achieved by abstract gabarge collection and pushdown flow
analysis individually.
The size of the abstract transition graph is a good heuristic measure of the temporal
reasoning ability of the analysis, \eg, its ability to support model-checking
of safety and liveness properties~\cite{mattmight:Might:2007:ModelChecking}.

In order to exercise both
well-known and newly-presented instances of CESK-based CFAs, we took a
series of small benchmarks exhibiting archetypal control-flow patterns
(see Figure~\ref{fig:table-results}).
Most benchmarks are taken from the CFA
literature: \texttt{mj09} is a running example from the work of
Midtgaard and Jensen designed to exhibit a non-trivial return-flow
behavior, \texttt{eta} and \texttt{blur} test common functional idioms,
mixing closures and eta-expansion, \texttt{kcfa2} and \texttt{kcfa3}
are two worst-case examples extracted from Van Horn and Mairson's proof of
$k$-CFA complexity~\cite{dvanhorn:VanHorn-Mairson:ICFP08},
\texttt{loop2} is an example from the Might's dissertation that
was used to demonstrate the impact of abstract GC~\cite[Section
  13.3]{mattmight:Might:2007:Dissertation}, \texttt{sat} is a
brute-force SAT-solver with backtracking.


\subsubsection{Comparing precision}
\label{sec:comparing-precision}

In terms of precision,
the fusion of pushdown analysis and abstract garbage collection
substantially cuts abstract transition graph sizes 
over one technique alone.

We also measure singleton flow sets as a heuristic metric for precision.
Singleton flow sets are a necessary precursor to optimizations
such as flow-driven inlining, type-check elimination and constant propagation.
Here again, the fused analysis prevails as the best-of- or better-than-both-worlds.

Running on the benchmarks, we have revalidated hypotheses about
the improvements to precision granted by both pushdown
analysis~\cite{dvanhorn:DBLP:conf/esop/VardoulakisS10} and abstract
garbage collection~\cite{mattmight:Might:2007:Dissertation}.
The table in Figure~\ref{fig:table-results} contains our detailed
results on the precision of the analysis. 


\subsubsection{Comparing speed}
\label{sec:comparing-speed}

In the original work on CFA2, Vardoulakis and Shivers present
experimental results with a remark that the running time of the
analysis is proportional to the size of the reachable
states~\cite[Section
6]{dvanhorn:DBLP:conf/esop/VardoulakisS10}. 
There is a similar correlation in the fused analysis,
but it is not as strong or as absolute.
From examination of the results, this appears to be because small graphs can have large stores inside each state,
which increases the cost of garbage collection (and thus transition) on a
per-state basis, and there is some additional per-transition overhead involved in maintaining the caches inside the Dyck state graph.
Table~\ref{fig:table-time} collects absolute execution
times for comparison.

\begin{figure}
\begin{center}
\begin{tabular}{|l|c@{\ \ }c|c@{\ \ }c|c@{\ \ }c|c@{\ \ }c|}
\hline
Program & 
\multicolumn{2}{c|}{$0$-CFA} &
\multicolumn{2}{c|}{$0$-PDCFA} &
\multicolumn{2}{c|}{$1$-CFA} &
\multicolumn{2}{c|}{$1$-PDCFA} \\

\hline \hline

\texttt{mj09} & 
$1''$ & $\epsilon$ &
$\epsilon$ & $1''$ &
$4''$ & $\epsilon$ &
$\epsilon$ & $\epsilon$ \\  
\hline

\texttt{eta} & 
$\epsilon$ & $\epsilon$ &
$\epsilon$ & $\epsilon$ &
$1''$ & $\epsilon$ &
$\epsilon$ & $\epsilon$ \\  
\hline

\texttt{kcfa2} & 
$1''$ & $\epsilon$ &
$\epsilon$ & $1''$ &
$24''$ & $\epsilon$ &
$1''$ & $\epsilon$ \\  
\hline

\texttt{kcfa3} & 
$2''$ & $\epsilon$ &
$\epsilon$ & $1''$ &
$\infty$ & $1''$ &
$58''$ & $2''$ \\  
\hline 

\texttt{blur} & 
$\infty$ & $7''$ &
$2'$ & $50''$ &
$4'$ & $30''$ &
$11''$ & $55''$ \\  
\hline

\texttt{loop2} & 
$36''$ & $1''$ &
$29''$ & $16''$ &
$\infty$ & $5''$ &
$13'$ & $2'$ \\  
\hline

\texttt{sat} & 
$\infty$ & $45''$ &
$6'$ & $19'$ &
$\infty$ & $3'$ &
$12'$ &  $37''$ \\  
\hline

\end{tabular}
\end{center}

\caption{
We ran our benchmark suite on a 2 Core 2.66 GHz OS~X machine with
4 Gb RAM. For each of the four analyses the left column denotes the
values obtained with no abstract collection, and the right one---with
GC on. The results of the analyses are presented in minutes ($'$) or
seconds ($''$), where $\epsilon$ means a value less than 1 second and
$\infty$ stands for an analysis, which has been interrupted due to the
an execution time greater than 30 minutes.
  }
\label{fig:table-time}
\end{figure}

\begin{figure*}
\centering
\includegraphics[scale=0.6]{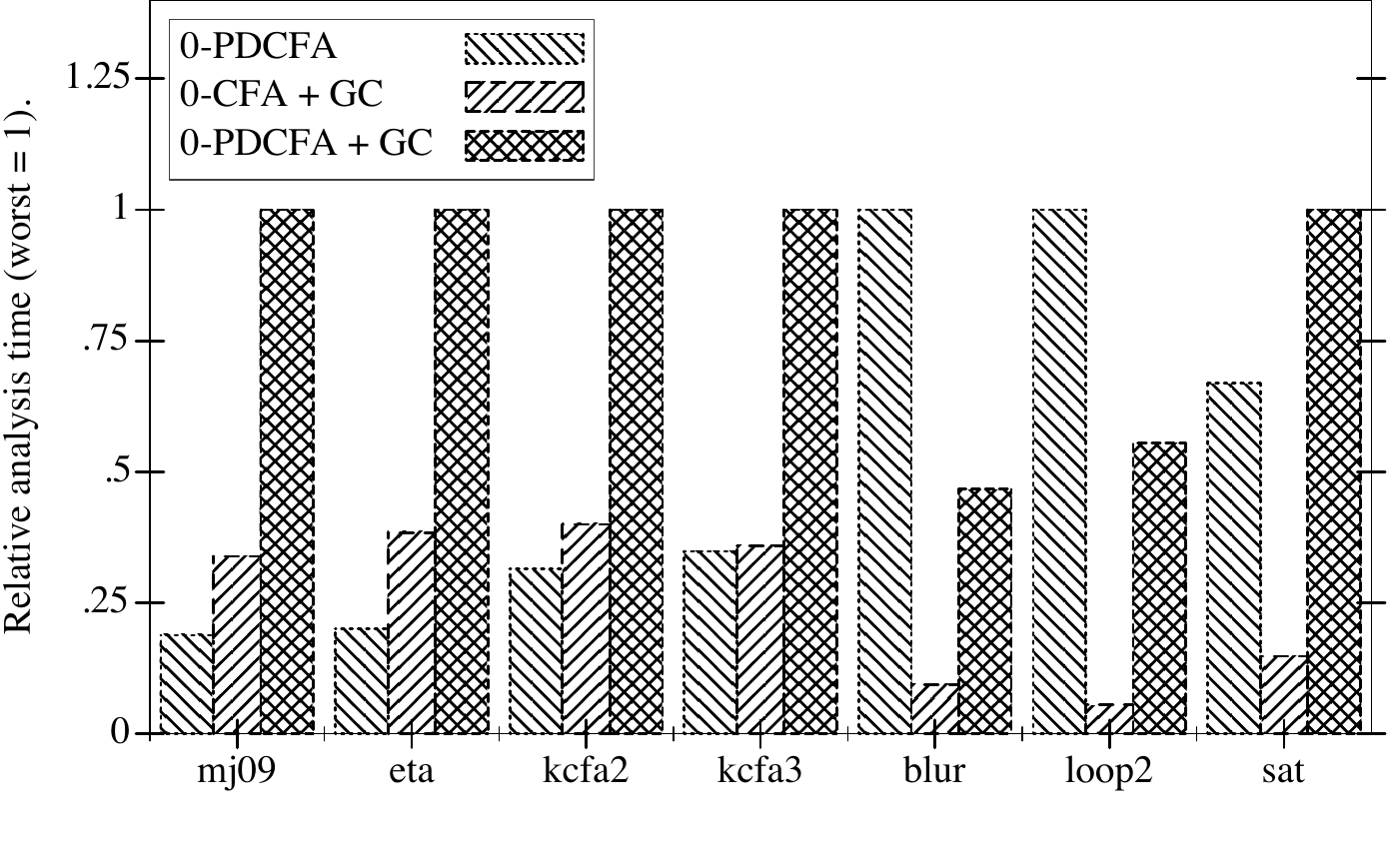}
\hfill
\includegraphics[scale=0.6]{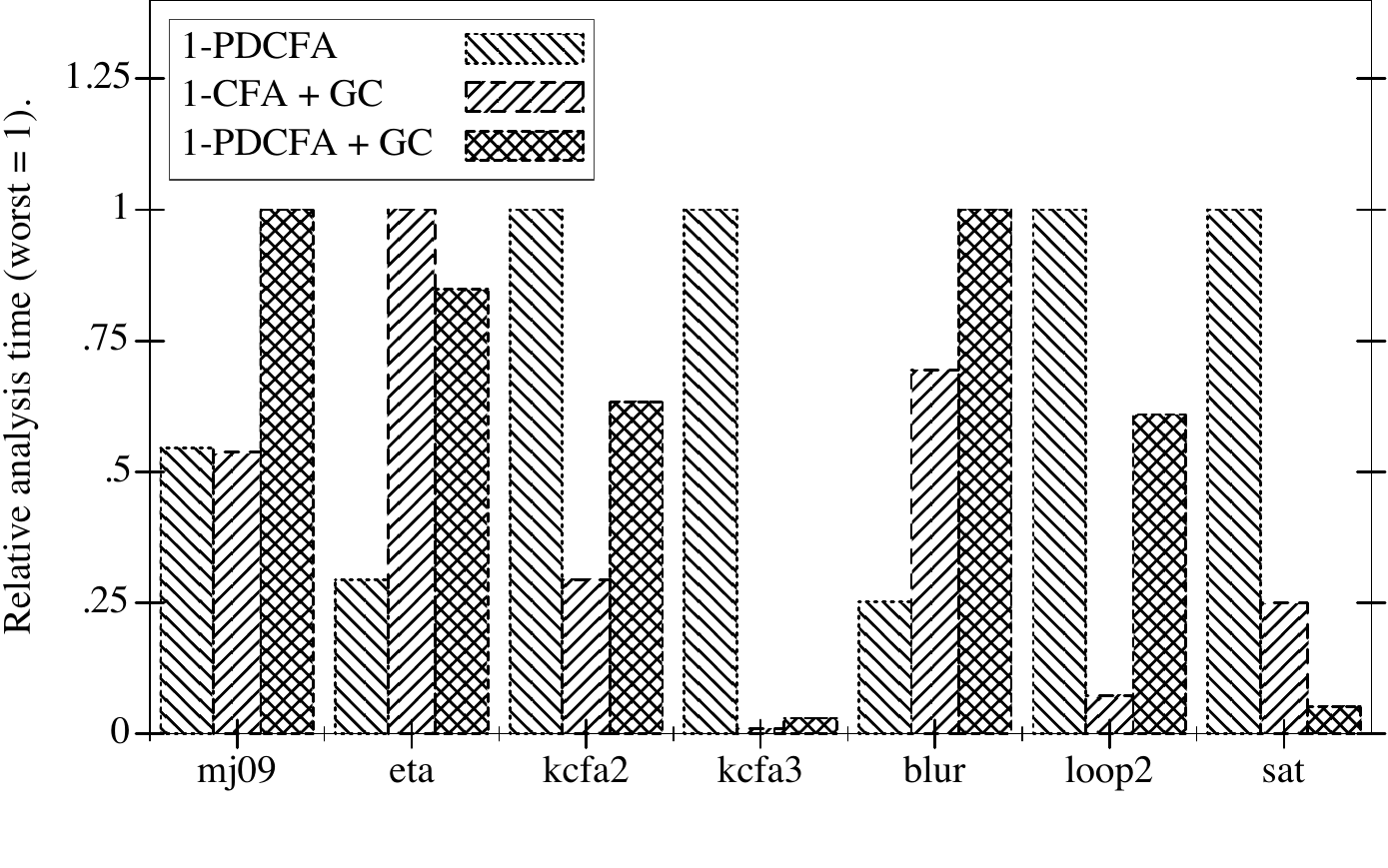}
\caption{Analysis times relative to worst (= 1) in class; smaller is
  better.  On the left is the monovariant 0CFA class of analyses, on
  the right is the polyvariant 1CFA class of analyses. (Non-GC
  $k$-CFA omitted.)}
\label{fig:execution-times}
\end{figure*}

It follows from the results that pure machine-style $k$-CFA is always
significantly worse in terms of execution time than either with GC or push-down
system. The histogram on Figure~\ref{fig:execution-times} presents normalized
relative times of analyses' executions.
About half the time, the fused analysis is faster than one of pushdown analysis
or abstract garbage collection.
And about a tenth of the time, it is faster than both.\footnote{
  The SAT-solving bechmark showed a dramatic improvement with the addition
  of context-sensitivity.
  Evaluation of the results showed that context-sensitivity provided enough 
  fuel to eliminate most of the non-determinism from the analysis.
}
When the fused analysis is slower than both, it is generally not much worse than
twice as slow as the next slowest analysis.

Given the already substantial reductions in analysis times provided by
collection and pushdown anlysis, the amortized penalty is a small and
acceptable price to pay for improvements to precision.


\section{Related work}
\label{sec:related}

Garbage-collecting pushdown control-flow analysis draws on work in higher-order
control-flow analysis~\cite{mattmight:Shivers:1991:CFA}, abstract
machines~\cite{mattmight:Felleisen:1987:CESK} and abstract
interpretation~\cite{mattmight:Cousot:1977:AI}.

%
%
%
%
%
%

\paragraph{Context-free analysis of higher-order programs}
The motivating work
for our own is Vardoulakis and Shivers very recent discovery of
  CFA2~\cite{dvanhorn:DBLP:conf/esop/VardoulakisS10}.
CFA2 is a table-driven summarization algorithm that exploits the balanced
nature of calls and returns to improve return-flow precision in a control-flow
analysis.
Though CFA2 exploits context-free languages, context-free languages are not
explicit in its formulation in the same way that pushdown systems are explicit in
our presentation of pushdown flow analysis.
With respect to CFA2, our pushdown flow analysis is also polyvariant/context-sensitive (whereas CFA2 is monovariant/context-insensitive), and it
covers direct-style.

On the other hand, CFA2 distinguishes stack-allocated and
store-allocated variable bindings, whereas our formulation of pushdown
control-flow analysis does not: it allocates all bindings in the
store.
If CFA2 determines a binding can be allocated on the stack, that
binding will enjoy added precision during the analysis and is not
subject to merging like store-allocated bindings.
While we could incorporate such a feature in our formulation,
it is not necessary for achieving ``pushdownness,''
and in fact, it could be added to classical finite-state CFAs as well.

\paragraph{Calculation approach to abstract interpretation}

Midtgaard and Jensen~\cite{dvanhorn:Midtgaard2009Controlflow} 
systematically calculate
0CFA using the Cousot-Cousot-style calculational approach to abstract
interpretation~\cite{dvanhorn:Cousot98-5} applied to an ANF
$\lambda$-calculus.
Like the present work, Midtgaard and Jensen start with the CESK
machine of Flanagan~\emph{et al.}~\cite{mattmight:Flanagan:1993:ANF} and employ a
reachable-states model. 

The analysis is then constructed by composing well-known
Galois connections to reveal a 0CFA incorporating reachability.
The abstract semantics approximate the control stack component of the
machine by its top element.
The authors remark monomorphism materializes in two mappings: ``one
mapping all bindings to the same variable,'' the other ``merging all
calling contexts of the same function.''
Essentially, the pushdown 0CFA of Section~\ref{sec:abstraction}
corresponds to Midtgaard and Jensen's analysis when the latter
mapping is omitted and the stack component of the machine is not
abstracted.

\paragraph{CFL- and pushdown-reachability techniques}
This work also draws on CFL- and pushdown-reachability
analysis~\cite{mattmight:Bouajjani:1997:PDA-Reachability,mattmight:Kodumal:2004:CFL,mattmight:Reps:1998:CFL,mattmight:Reps:2005:Weighted-PDA}.
For instance, \ecg s, or equivalent variants thereof, appear in many
context-free-language and pushdown reachability algorithms.
For our analysis, we implicitly invoked
these methods as subroutines.
When we found these algorithms lacking (as with their enumeration of
control states), we developed Dyck state graph construction.

CFL-reachability techniques have also been used to compute classical
finite-state abstraction CFAs~\cite{mattmight:Melski:2000:CFL} and
type-based polymorphic control-flow
analysis~\cite{mattmight:Rehof:2001:TypeBased}.
These analyses should not be confused with pushdown control-flow
analysis, which is computing a fundamentally more precise kind of CFA.
Moreover, Rehof and Fahndrich's method is cubic in the size of the
\emph{typed} program, but the types may be exponential in the size of
the program.
Finally, our technique is not restricted to typed programs.

\paragraph{Model-checking higher-order recursion schemes}
There is terminology overlap with work by
Kobayashi~\cite{mattmight:Kobayashi:2009:HORS} on model-checking higher-order
programs with higher-order recursion schemes, which are a
generalization of context-free grammars in which productions can take
higher-order arguments, so that an order-0 scheme is a context-free
grammar.
Kobyashi exploits a result by Ong~\cite{dvanhorn:Ong2006ModelChecking} which
shows that model-checking these recursion schemes is decidable (but
ELEMENTARY-complete) by transforming higher-order programs into
higher-order recursion schemes.

Given the generality of model-checking, Kobayashi's technique may be
considered an alternate paradigm for the analysis of
higher-order programs.
For the case of order-0, both Kobayashi's technique and our own
involve context-free languages, though ours is for control-flow
analysis and his is for model-checking with respect to a temporal
logic.
After these surface similarities, the techniques diverge.
In particular, higher-order recursions schemes are limited
to model-checking programs in the simply-typed 
lambda-calculus with recursion.







\section{Conclusion}

Our motivation was to further probe the limits of decidability
for pushdown flow analysis of higher-order programs
by enriching it with abstract garbage collection.
We found that abstract garbage collection broke
the pushdown model, but not irreparably so.
By casting abstract garbage collection in terms of
an introspective pushdown system and synthesizing
a new control-state reachability algorithm, we have 
demonstrated the decidability of fusing two
powerful analytic techniques.

As a byproduct of our formulation, it was also easy
to demonstrate how polyvariant/context-sensitive 
flow analyses generalize to a pushdown formulation,
and we lifted the need to transform to continuation-passing style
in order to perform pushdown analysis.

Our empirical evaluation is highly encouraging: it shows that the fused
analysis provides further large reductions in the size of the abstract
transition graph---a key metric for interprocedural control-flow precision.
And, in terms of singleton flow sets---a heuristic metric for optimizability---the fused 
analysis proves to be a ``better-than-both-worlds'' combination.

Thus, we provide a sound, precise and polyvariant introspective
pushdown analysis for higher-order programs.

\acks
%
We thank our anonymous reviewers for their detailed comments on the submitted
paper. 
This material is based on research sponsored by DARPA under the programs 
Automated Program Analysis for Cybersecurity (FA8750-12-2-0106) and 
Clean-Slate Resilient Adaptive Hosts (CRASH). 
The U.S. Government is authorized to reproduce and
distribute reprints for Governmental purposes notwithstanding any copyright
notation thereon.


\bibliographystyle{acm}
\bibliography{mattmight,dvanhorn,local}




\end{document}